\documentclass[conference]{IEEEtran}
\usepackage{graphicx} 
\usepackage{algpseudocode}
\usepackage{algorithm}
\usepackage{bbm}
\usepackage{amsthm}
\usepackage{amsmath}
\newtheorem{prop}{Proposition}[section]
\usepackage{subfigure}

\newtheorem{theorem}{Theorem}[section]
\usepackage{algpseudocode}
\usepackage{algorithm}
\usepackage{multirow}
\usepackage{booktabs}  
\usepackage{footnote}
\usepackage{url}
\ifCLASSINFOpdf
\else
\fi
\begin{document}
\title{SECO: Secure Inference With Model Splitting \\ Across Multi-Server Hierarchy}

% \author{Shuangyi Chen \thanks{Department of Electrical and Computer Engineering, University of Toronto, shuangyi.chen@mail.utoronto.ca}, Ashish Khisti \thanks{Department of Electrical and Computer Engineering, University of Toronto, akhisti@ece.utoronto.ca}}
% \author{Shuangyi Chen \textsuperscript{\textsection}, Ashish Khisti \textsuperscript{\textsection}}

% \author{
%   \IEEEauthorblockN{Shuangyi Chen\IEEEauthorrefmark{1},
%                     and Ashish Khisti \IEEEauthorrefmark{1}}
%   \IEEEauthorblockA{\IEEEauthorrefmark{1}
%   \footnotetext{Department of Electrical and Computer Engineering,
%                     University of Toronto, 
%                     \{shuangyi.chen@mail.utoronto.ca, akhisti@ece.utoronto.ca\}}
%   }}
\author{%
  \IEEEauthorblockN{Shuangyi Chen\IEEEauthorrefmark{1},
                    and Ashish Khisti \IEEEauthorrefmark{1}}
  \IEEEauthorblockA{\IEEEauthorrefmark{1}%
  Department of Electrical and Computer Engineering,
                    University of Toronto, \\
                    \{shuangyi.chen@mail.utoronto.ca, akhisti@ece.utoronto.ca\}}
}
% \author{\IEEEauthorblockN{Shuangyi Chen}
% \IEEEauthorblockA{Department of Electrical and Computer Engineering\\
% University of Toronto}
% \and
% \IEEEauthorblockN{Ashish Khisti}
% \IEEEauthorblockA{Department of Electrical and Computer Engineering\\
% University of Toronto}
% }

% \author{%
%   \IEEEauthorblockN{Anonymous Authors}
% }



\maketitle

\begin{abstract}
In the context of prediction-as-a-service, concerns about the privacy of the data and the model have been brought up and tackled via secure inference protocols. These protocols are built up by using single or multiple cryptographic tools designed under a variety of different security assumptions.

In this paper, we introduce SECO, a secure inference protocol that enables a user holding an input data vector and multiple server nodes deployed with a split neural network model to collaboratively compute the prediction, without compromising either party's data privacy. We extend prior work on secure inference that requires the entire neural network model to be located on a single server node, to a multi-server hierarchy, where the user communicates to a  gateway server node, which in turn communicates to remote server nodes. The inference task is split across the server nodes and must be performed over an encrypted copy of the data vector.

We adopt multiparty homomorphic encryption and multiparty garbled circuit schemes, making the system secure against dishonest majority of semi-honest servers as well as protecting the partial model structure from the user. We evaluate SECO on multiple models, achieving the reduction of computation and communication cost for the user, making the protocol applicable to user's devices with limited resources.
\end{abstract}
\begin{IEEEkeywords}
Distributed machine learning, Security and privacy
\end{IEEEkeywords}

% \fi
\IEEEpeerreviewmaketitle

\section{Introduction}
Prediction-as-a-service using neural network models enables a company to deploy a neural network on the cloud and allows users to upload their input and obtain the corresponding prediction results. However, this may lead to privacy concerns and legal issues. There are regulatory constraints such as Health Insurance Portability and Accountability Act (HIPAA) and EU General Data Protection Regulation (GDPR), restricting the way data could be used or shared.  A secure inference protocol enables a user, holding its input, and servers deployed with the neural network to collaboratively evaluate the prediction result, preserving the privacy of both the input and the model parameters. Many cryptographic tools, such as Homomorphic Encryption (HE) \cite{Fan2012SomewhatPF} and secure multiparty computation (MPC), have been applied in this area. Homomorphic encryption allows the data holder to encrypt its data with the public key, any other party who obtains the encrypted data can compute on it without decryption, and only the party holding the secret key can decrypt the computed result. MPC is usually a mixed approach that combines secret sharing and Garbled Circuit \cite{Yao1986HowTG}, etc. Both HE and MPC can protect data privacy. Pure HE-based secure inference protocols \cite{gilad2016cryptonets}\cite{https://doi.org/10.48550/arxiv.1711.05189}\cite{dathathri2019chet}\cite{https://doi.org/10.48550/arxiv.1806.03461}  provides the data and model confidentiality as well as protects the structure of the model. However, it involves heavy computation and does not immediately generalize to non-linear function computations used during predictions.  One popular approach in HE-based protocols is to replace the activation functions in neural networks with approximated polynomials at the cost of reduced accuracy. While MPC protocols are not as computationally heavy as HE-based protocols, they require multiple rounds of communication and can also leak some information about the model to the users. Furthermore, there are hybrid approaches that combine HE and MPC. GAZELLE \cite{Gazelle} relies on a hybrid approach by employing HE and two-party computation, performing linear function evaluation with HE and non-linear function evaluation with Garbled Circuit. DELPHI \cite{Delphi} builds upon techniques from GAZELLE and modifies it to a two-phase protocol. It reduces the use of heavy cryptographic tools in the online phase, thus speeding up online inference. {However, these methods come with certain limitations. Firstly, they reveal the neural network's architecture to the user, which is problematic when the model's structure is intended to remain confidential, a common scenario with today's large language models. Secondly, they require user equipment to possess substantial computational power and communication capabilities. Additionally, they operate under the assumption that the model is hosted on a single server, which is impractical for very large models. For instance, processing a model with the size of GPT-4 necessitates a cluster of multiple GPUs. To process such a large model effectively, distributing the model parameters across several servers, each with limited GPU resources, can enhance processing power.}

\begin{figure}[h]
    \centering
    \includegraphics[scale=0.52]{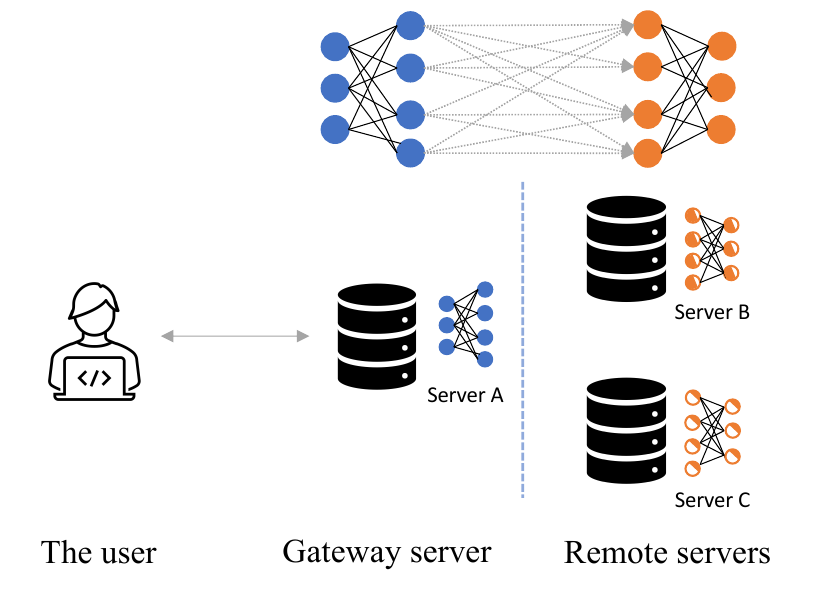} 
    \caption{Protocol Architecture}
    \label{fig:arch}
\end{figure}

With these considerations, we propose SECO, a hybrid HE-MPC secure inference protocol that involves a hierarchy of servers as shown in Figure~\ref{fig:arch}. The user communicates with a {\em gateway} server that stores a part (the parameters associated with the first few layers) of the neural network model. The remainder of the neural network model is stored on remote servers, with each server holding a share of the weights of the linear layers. Our setup extends the secure inference protocol proposed in~\cite{Delphi} to such a multi-server scenario, which has several advantages. First, the user only needs to participate during the evaluation involving part of the network model stored on the gateway server. 
It does not need to participate during evaluations involving the network model stored on the remote servers. This reduces the communication and computational requirements for the user device while increasing the load on the servers, {which enables the use of SECO on lightweight user devices.}
% By controling the number of neural network layers stored in the gateway server and remote servers, we can provide a tradeoff between the computational and communication load at the user and server nodes. 
Furthermore although as in prior works~\cite{Delphi}, the structures of the layers in the neural network model stored on the gateway server must be revealed to the user, the structures of the layers on the remote server are not revealed to the user. {SECO offers the benefit of hiding most of the neural network's architecture, crucial for protecting the proprietary structures of large models like GPT-4, where both the model's parameters and the model's design are valuable assets.} We demonstrate that a straightforward application of prior works to the hierarchical server setting can lead to information leakage. Our proposed setting involves storing a secret share of model weights at the remote server. However, for secure inference of the partial model on the remote servers, a straightforward extension of DELPHI's 2PC protocol to three-party setting with distributed weights is not privacy-preserving due to the limitation of the plain HE scheme. To privately perform the execution of secure inference on the remote servers, we design a sub-protocol relying on multiparty HE, secret sharing, and multiparty Garbled Circuit schemes, {ensuring the privacy of each individual data while facilitating efficient communication. We achieve this by 1) designing new secret sharing scheme and integrating it with multiparty HE for secure computation of linear functions 2) departing from the conventional approach for the use of Garbled Circuits (GC) in DELPHI by redesigning the assignment of roles for secure computation of non-linear functions.} We formally demonstrate that our proposed method prevents information leakage in the presence of dishonest majority of honest-but-curious servers. 
We have implemented SECO and reported the performance. The results indicate that SECO significantly enhances efficiency by reducing the user's online participation time by up to $9.9\times$ and decreasing the user's online communication costs by as much as $18.5\times$ compared to DELPHI. It also provides a versatile tradeoff between the communication and computation loads between the user and server in addition to protecting more information about the neural network than prior works~\cite{Delphi}.

\section{Related Work}
\subsection{Privacy-Preserving Inference on Neural Networks}
Secure Inference has received considerable attention in recent years. It is developed based on the following techniques: homomorphic encryption \cite{gentry2009fully}, secret sharing and garbled circuits \cite{Yao1986HowTG}. Some earlier works utilize a single technique and their models are simple machine learning models such as linear regression \cite{6547119} \cite{wu2012using}, linear classifiers \cite{bost2014machine} \cite{bos2014private} \cite{graepel2012ml}, etc.

\paragraph{HE-based Protocols} Pure HE-based protocols usually protect the complete model information including the architecture but lack support for non-approximated activation functions. CryptoNets \cite{gilad2016cryptonets} was one of the earliest works to use homomorphic encryption on neural network inference. CryptoDL \cite{https://doi.org/10.48550/arxiv.1711.05189} improves upon CryptoNets by choosing the interval and the degree of approximations based on heuristics on the data distribution. CHET \cite{dathathri2019chet} also improves upon CryptoNets by modifying the approximated activation function that CryptoNets uses to $f(x) = ax^2 +bx$ with learnable parameters $a$ and $b$. TAPAS \cite{https://doi.org/10.48550/arxiv.1806.03461} utilizes FHE \cite{chillotti2016faster} to allow arbitrary-layer neural network processing and proposes tricks to accelerate computation on encrypted data. nGraph-HE \cite{https://doi.org/10.48550/arxiv.1810.10121} develops HE-aware graph-compiler optimizations on general matrix-matrix multiplication operations and convolution-batch-norm operation. CHIMERA \cite{Boura2018ChimeraAU} uses different HE schemes for ReLU functions and other neural network functions. It needs to switch between two HE schemes to evaluate the model, leading to high computational costs. Cheetah \cite{9407118} presents a set of optimizations to speed up HE-based secure inference, including auto-tuning HE security parameters, the order scheduler for HE operations, and a hardware accelerator architecture. 
Pure HE-based protocols involve heavy computation and their performance are not practical  due to their costly homomorphic computations. 

\paragraph{MPC-based Protocols}
Pure MPC-based protocols are usually high-throughput but with high communication overhead due to multiplication operations. Furthermore, they usually need to expose the structure of the model to all parties. ABY \cite{demmler2015aby} is a 2PC protocol that supports arbitrary neural network functions by combining Arithmetic, Boolean circuits, and Yao’s garbled circuits. ABY3 \cite{mohassel2018aby3} improves ABY by extending it to 3-party settings and optimizing the conversion between Arithmetic, Boolean circuits, and garbled circuits. {3PC protocol SecureNN \cite{wagh2019securenn} distributes 2-out-of-2 secret sharing of the input among two of three servers. However, it is only demonstrated on small DNNs.  
FALCON, a 3PC protocol proposed by Wagh et al.\cite{wagh2020falcon} uses replicated secret sharing to reduce the number of interactions. However, the use of replicated secret sharing in the context of prediction-as-a-service necessitates an increase in communication costs for the user.
% Protocols including FLASH \cite{byali2019flash}, SWIFT \cite{https://doi.org/10.48550/arxiv.2005.10296}, Trident \cite{chaudhari2019trident}, Fantastic Four \cite{cryptoeprint:2020/1330}, Tetrad \cite{koti2021tetrad} are in the 4PC setting, assuming honest majority which is a strong assumption of the system. 
To sum up, MPC-based protocols rely on assumptions such as non-collusion, or an honest majority among the servers. Such a strong assumption of the system might not always align with the practical environments. Conversely, SECO assumes dishonest majority among servers, meaning the user's input privacy can be protected even when most of the servers are corrupted by the adversary. Thus, SECO can offer a higher degree of trustworthiness to users than pure MPC-based protocols. Furthermore, these 3PC protocols require the user to divide the input into multiple random shares. Each input share is then dispatched to three different servers, thereby increasing the communication costs for the user. }

\paragraph{HE-MPC-based Protocols}
HE-MPC-based Protocols are designed by using both HE and MPC in a hybrid scheme. GAZELLE \cite{Gazelle} designs schemes for HE neural network operations. It uses HE to evaluate linear layers and Garbled Circuits to evaluate non-linear layers, achieving efficient secure inference with high accuracy. DELPHI \cite{Delphi} modifies GAZELLE by designing two phases with the offline phase using HE, speeding up the online inference by evaluating the neural network in plaintext. It also proposes a novel neural network architecture search planner to determine where to approximate ReLU in the neural network, trading off between performance and accuracy. MP2ML \cite{10.1145/3411501.3419425} is proposed based on a novel combination of nGraph-HE \cite{https://doi.org/10.48550/arxiv.1810.10121} and ABY \cite{demmler2015aby}, combining additive secret sharing and CKKS \cite{cheon2017homomorphic} HE scheme. {However, HE-MPC-based Protocols face several limitations. Firstly, they require revealing the neural network's architecture to users, problematic for confidential models like large language models. Secondly, they demand substantial computational power and communication capabilities from user equipment, which can be a barrier for those with limited resources. Thirdly, these methods assume the model is hosted on a single server, impractical for large models like GPT-4 that need multiple GPUs. SECO is derived from these 2PC HE-MPC-based Protocols and is optimized to address the mentioned limitations.}

\subsection{Split Neural Network}
{
Split neural network is a specific type of neural network architecture. In this method, the neural network is divided into two or more sections, and each section is processed on different devices or servers. It is often used to reduce load for devices with limited resources. SplitFed \cite{thapa2022splitfed} is a federated learning framework that combines federated learning with split learning for better model privacy. However, it applies differential privacy rather than cryptographic tools. 
Pereteanu et al. introduce Split HE \cite{pereteanu2022split}, a framework where the server provides the user with the middle portion of the neural network. In this setup, the user conducts plaintext inference on the middle section of the network, while the server manages the inference of other parts under homomorphic encryption. However, this approach harms the privacy of the deployed model.
Khan et al. propose a secure training protocol named Split Ways \cite{khan2023split} based on U-shaped split learning, which incorporates homomorphic encryption on the client side to protect user input privacy. However, while they deploy only the fully-connected layer on the server, the user is still responsible for training the majority of the neural network. 
Split HE and Split Ways explore the integration of split neural networks with cryptographic tools. However, they do not fully utilize the server's computing resources and continue to depend significantly on the user's device.
}

\section{Preliminaries}
% \subsection{Notations}
%  We provide Table~\ref{tab:tab1} for notation summary.
% \begin{table}
% \centering
% \caption{Notation Summary}\label{tab:tab1}
% \begin{tabular}{c|c}
% \toprule
% Symbol & Description \\
% \midrule
% $j$ & Party index \\
% % $\leftarrow$& uniform sampling\\
% $\mathbbm{Z}_{q}$ & $[-\frac{q}{2},\frac{q}{2})$\\
% $R_t$ & Plaintext space for HE scheme\\
% $R_q$ & Ciphertext space for HE scheme \\
% $\leftarrow$ & uniform sampling  \\
% $L$ & Total number of layers of the model  \\
% $l$ & The number of layers stored on server A \\
% $q_{i}$ & Input size of $i^{\text{th}}$ layer\\
% $w_{i}$ & Output size of $i^{\text{th}}$ layer\\
% $\mathbf{F}_{A}$ & The partial model on server A\\
% $\mathbf{F}_{BC}$ & The partial model on server B and C\\
% $\mathbf{F}_{i}$ & Model parameter of the $i^{\text{th}}$ layer\\
% $\mathbf{F}_{i}^{j}$ & Model parameter of the $i^{\text{th}}$ layer on Party $j$\\
% $\mathbf{x}_{i}$ & The input of the model's $i^{\text{th}}$ layer\\
% $\hat{\mathbf{y}}$ & The prediction result\\
% $\mathbf{r}_{i}$ & Randomness masking the input of $i^{\text{th}}$ layer\\
% $\mathbf{s}_{i}$ & Randomness masking the output of $i^{\text{th}}$ layer\\
% $(\textsf{sk}_{j},\textsf{pk}_{j})$ & Key pair generated by party $j$  \\
% $\textsf{cpk}$& Common public key\\
% $ct_x$ & Ciphertext of message $x$ \\
% $C_i$ & Circuit for computing the $i^{\textsf{th}}$ ReLU layer \\
% $\Tilde{C}_i$ &  Garbled circuit correspongind to $C_i$\\
% $\textsf{label}_x$ & GC input labels corresponding to value $x$ \\

% \bottomrule
% \end{tabular}
% \end{table}
\subsection{Split Neural Network}
\label{section:splitnn}
Split Neural Network \cite{vepakomma2018split}\cite{gupta2018distributed} involves splitting a neural network model at different intermediate layers. Consider a neural network function $\mathbf{F}$ composed of a sequence of layers $\{\mathbf{F}_1,...,\mathbf{F}_L\}$. Each layer $\mathbf{F}_i$ can be a linear layer or an activation layer. The prediction result of the model for a given input $\mathbf{x}_1$ can be computed by $\mathbf{F}(\mathbf{x}_1) = \mathbf{F}_L(\mathbf{F}_{L-1}(...(\mathbf{F}_{1}(\mathbf{x}_1))))$. We split the model at the $l^{\text{th}}$ layer to get two models $\mathbf{F}_{\uppercase\expandafter{\romannumeral1}} = \{\mathbf{F}_1,...,\mathbf{F}_l\}$, $\mathbf{F}_{\uppercase\expandafter{\romannumeral2}}  = \{\mathbf{F}_{l+1},...,\mathbf{F}_L\}$ which could be stored at two different nodes. The output can be expressed as $\mathbf{F}(\mathbf{x}_1) = \mathbf{F}_{\uppercase\expandafter{\romannumeral2}}(\mathbf{F}_{\uppercase\expandafter{\romannumeral1}}(\mathbf{x}_1))$. We adjust $l$ to adjust the processing (training or inference) load for the parties deployed with $\mathbf{F}_{\uppercase\expandafter{\romannumeral1}}$ and $\mathbf{F}_{\uppercase\expandafter{\romannumeral2}}$.
\subsection{Cryptographic Blocks}\label{sec:cryptographic-blocks}
We use the following cryptographic blocks to build the protocol.
\subsubsection{The Plain BFV Homomorphic Encryption} \label{plainBFV}
The Brakerskil-Fan-Vercauteren scheme \cite{Fan2012SomewhatPF} is a Ring-Learning with Errors (RLWE)-based cryptosystem that supports additive and multiplicative homomorphic operations. We define the plaintext space to be $R_t = \mathbbm{Z}_t [X]/(X^{n}+1)$ and the ciphertext space to be $R_q = \mathbbm{Z}_q [X]/(X^{n}+1)$, where $n$ is a power of 2. We denote $\Delta = \lceil \frac{q}{t}\rceil$. BFV Homomorphic Encryption scheme can be expressed as a tuple of functionalities $\textsf{HE}=(\textsf{HE.KeyGen}, \textsf{HE.Enc}, \textsf{HE.Dec}, \textsf{HE.Eval})$. The scheme is based on two kinds of distributions: the key distribution $R_3 = \mathbbm{Z}_3 [X]/(X^{N}+1)$ with coefficients uniformly distributed in $\{-1,0,1\}$; the error distribution $\chi$ over $R_q$ with coefficients distributed according to a centered discrete Gaussian.
\begin{itemize}
    \item $\textsf{HE.KeyGen(\text{Params})}\rightarrow (\textsf{pk},\textsf{sk})$:It takes a security parameter as input and returns a public key $\textsf{pk}$ and a secret key $\textsf{sk}$.
    \begin{algorithm}
\caption{$\textsf{HE.KeyGen(params)} \rightarrow (\textsf{pk},\textsf{sk})$} 
\label{alg:keygen}
{\small{
    \begin{algorithmic}[1]
    \State Sample $s \leftarrow R_3$
    \State Let $\textsf{sk} = s$. Sample $p_1 \leftarrow R_q$, $e\leftarrow \chi$
    \State $\textsf{pk} = (p_0,p_1)=(-s p_1+e,p_1)$. Output $(\textsf{pk},\textsf{sk})$
    \end{algorithmic}}}
\end{algorithm}
    \item $\textsf{HE.Enc}(\textsf{pk}, m)\rightarrow ct$: It takes as input the public key $\textsf{pk}$ and message $m$, and outputs the ciphertext of $m$ denoted as $ct$.
     \begin{algorithm}
\caption{$\textsf{HE.Enc}(\textsf{pk}, m) \rightarrow ct$} 
\label{alg:enc}
{\small{
    \begin{algorithmic}[1]
    \State Let $\textsf{pk} = (p_0,p_1)$. Sample $u\leftarrow R_3$, $e_0,e_1\leftarrow \chi$
    \State Output  $ct = (\Delta m +u p_0 +e_0, u p_1+ e_1)$
    \end{algorithmic}}}
\end{algorithm}
\item $\textsf{HE.Dec}(\textsf{sk}, ct)\rightarrow m$: It takes as input the secret key $\textsf{sk}$ and the ciphertext $ct$, and outputs message $
    m$.
    \begin{algorithm}
    % \renewcommand{\thealgorithm}{}
    % \floatname{algorithm}{}
    \caption{$\textsf{HE.Dec}(\textsf{sk}, ct) \rightarrow m$} 
    \label{alg:dec}
    {\small{
    \begin{algorithmic}[1]
    \State Let $\textsf{sk} = s$, $ct =(c_0,c_1)$
    \State Output  $m = [\lfloor\frac{t}{q}[c_0+c_1 s]_q \rceil]_t$
    \end{algorithmic}}}
\end{algorithm}
\item $\textsf{HE.Eval}(\textsf{pk},\{ct_{i},pt_{i}\},f)\rightarrow ct_{f}$: It takes as input the public key $\textsf{pk}$, valid ciphertexts $\{ct_{i}\}$ encrypting $\{m_{i}\}$ or plaintexts $\{pt_{i}\}$ encoding $\{m_{i}\}$ and operation $f$, and returns a valid ciphertext $ct_{f}$ encrypting $f(\{m_{i}\}_i)$. The class of function $f(\cdot)$ that are supported include element-wise addition \textsf{Add}, subtraction \textsf{Sub} and multiplication \textsf{Mul} and slot rotation \textsf{Rot}, as well as a combination of those basic operations.
\end{itemize}
\subsubsection{Multiparty BFV Homomorphic Encryption (MPHE)}
The Multiparty version of BFV Homomorphic Encryption scheme (MPHE) \cite{mouchet2021multiparty} is extended from the plain BFV scheme. It enables multiple distributed parties, each with input in private, to collaboratively compute a function of those inputs without leaking any party's input. In this scheme, a common public key collectively generated by parties is known to all parties, while the corresponding secret key $\textsf{csk}$ is distributed among parties. Below we introduce the main functions used in the protocol. Consider a scenario that $N$ parties want to collaboratively compute a function of their private inputs.
\begin{itemize}
% [itemsep=0.5pt,topsep=0.8pt,parsep=0.5pt]
    \item $\textsf{MPHE.KeyGen}(\textsf{params})\rightarrow (\textsf{pk}_{i},\textsf{sk}_{i})$ :Each party runs $\textsf{HE.KeyGen}(\textsf{params})$ and generates a public key $\textsf{pk}_{i}$ and secret key $
    \textsf{sk}_{i}$. This procedure requires a public polynomial $p_1$, which is agreed upon by all $N$ parties.
    \item $\textsf{MPHE.DKeyGen}(\{\textsf{pk}_{i}\}_{i=0}^{N-1})\rightarrow \textsf{cpk}$ :It returns a common public key $\textsf{cpk}$, for a set of public keys $\{\textsf{pk}_{i}\}_{i=1}^{N-1}$. Its corresponding secret key $\textsf{csk}$ is the aggregation of $\{\textsf{sk}_{i}\}_{i=1}^{N-1}$.
    \begin{algorithm}
    % \renewcommand{\thealgorithm}{}
    % \floatname{algorithm}{}
    \caption{$\textsf{MPHE.DKeyGen}(\{\textsf{pk}_{i}\}_{i=0}^{N-1}) \rightarrow \textsf{cpk}$} 
    \label{alg:dkeygen}
    {\small{
    \begin{algorithmic}[1]
    \State Let $\textsf{pk}_i = (p_{0,i},p_1)$
    \State Output  
    $\textsf{cpk} = ([\sum_{i=0}^{N-1} p_{0,i}]_q,p_1)$
    \end{algorithmic}}}
\end{algorithm}
    \item $\textsf{MPHE.Enc}(\textsf{cpk}, m)\rightarrow ct$ :It returns a ciphertext $ct$ by running $\textsf{HE.Enc}(\textsf{cpk},m)$.
    \item $\textsf{MPHE.Eval}(\textsf{cpk},\{ct_{i},pt_{i}\},f)\rightarrow ct_{f}$ :It returns a ciphertext $ct_{f}$ by running $\textsf{HE.Eval}(\textsf{cpk},\{ct_{i},pt_{i}\},f)$.
    \item $\textsf{MPHE.Reconstruct}(ct, \textsf{sk}_{i})\rightarrow pd_{i}$ :Used by each party $i$. It takes as input a publicly known ciphertext $ct$, party's secret key $\textsf{sk}_{i}$, outputs a partial decryption $pd_{i}$.
    \begin{algorithm}
    % \renewcommand{\thealgorithm}{}
    % \floatname{algorithm}{}
    \caption{$\textsf{MPHE.Reconstruct}(ct, \textsf{sk}_{i}) \rightarrow pd_i$} 
    \label{alg:reconstruct}
    {\small{
    \begin{algorithmic}[1]
    \State Let $\textsf{sk}_i = s_i, ct = (c_0,c_1)$. Sample $e_i \leftarrow \chi$
    \State Output $pd_i = s_i c_1+ e_i$
    \end{algorithmic}}}
\end{algorithm}
    \item $\textsf{MPHE.Dec}(\{pd_{i}\}_{i=0}^{N-1},ct)\rightarrow m'$: It takes as input a set of partial decryption $\{pd_{i}\}_{i=0}^{N-1}$ and the corresponding publicly known ciphertext $ct$. By aggregating the partial decryption, it outputs $m'$ which is equal to $
    \textsf{HE.Dec}(ct,\textsf{csk})$, provided that the noise powers are again selected in a suitable manner.
     \begin{algorithm}
    % \floatname{algorithm}{}
    \caption{$\textsf{MPHE.Dec}(\{pd_{i}\}_{i=0}^{N-1},ct)\rightarrow m'$} 
    \label{alg:mphe-dec}
    {\small{
    \begin{algorithmic}[1]
    \State Let $ct = (c_0,c_1)$
    \State Compute 
    \begin{equation}
        c_0+c_1 s = c_0+ \sum_{i=0}^{N-1}pd_i = c_0+ c_1(\sum_{i=0}^{N-1} s_i)+ \sum_{i=0}^{N-1} e_i \nonumber
    \end{equation}
    % c_0+s c_1 = c_0+ \sum_{i=0}^{N-1}pd_i = c_0+ (\sum_{i=0}^{N-1} s_i)c_1+ \sum_{i=0}^{N-1} e_i$
    \State Output $m' = [\lfloor\frac{t}{q}[c_0+c_1 s]_q \rceil]_t$
    \end{algorithmic}}}
\end{algorithm}

\end{itemize}
\subsubsection{Oblivious Transfer (OT)}
Oblivious transfer \cite{Rabin2005HowTE} \cite{RandomizedProtocol} is a protocol involving two parties, a sender who has an input two messages $m_{0},m_{1}$, and a receiver who holds a bit $b$ as input. Via OT, the receiver obtains $m_b$. The security of the protocol guarantees that the sender does not learn anything about bit $b$ and the receiver does not learn anything about the message $m_{1-b}$.
\subsubsection{Garbled Circuits}
A Garbled Circuit scheme \cite{Yao1986HowTG} \cite{FoundationsofGC} is a cryptographic tool involving two parties joint computing a function $C(x_1,x_2)$ where $x_1$ and $x_2$ are inputs provided by the garbler and the evaluator. The scheme should keep the inputs fully private. It is consists of a tuple of algorithms \textsf{GC}=(\textsf{GC.Garble, GC.Transfer, GC.Eval}) with the following syntax:
\begin{itemize}
    \item $\textsf{GC.Garble}(\textsf{Params}, C)\rightarrow (\Tilde{C},\textsf{label)}$: Used by the garbler. Assume $C$ is composed of $G$ gates (e.g., XOR, AND, etc.) in total. The garbler garbles the circuit by 1) Assigning two random $k$-bit labels to each wire in the circuit $C$ corresponding to $0$ and $1$ where $k$ is the security parameter usually set to 128; 2) For each gate $g, g\in [G]$ in $C$, computing a garbled table $\Tilde{C}_g$ with $4$ rows corresponding to $4$ combinations of inputs labels. Each row is the encryption of the output with two corresponding input labels as the encryption key \cite{6547128}. Output $\Tilde{C} = \{\Tilde{C}_g\}_{g=0}^{G-1}$ and labels corresponding to input wires $\textsf{label}$.
    \item $\textsf{GC.Transfer}(\textsf{label}, x_i) \rightarrow$ $ \textsf{label}_{x_i}$ : This function involves both the garbler and the evaluator. If $i=2$, for input $x_2$ held by the garbler, the garbler maps $x_2$ with the labels generated for input wire corresponding to $x_2$ to obtain $\textsf{label}_{x_2}$, then sends $\textsf{label}_{x_2}$ to the evaluator. If $i=1$, for input $x_1$ held by the evaluator, the evaluator cannot simply send $x_1$ to the garbler to obtain $\textsf{label}_{x_1}$. This undermines input privacy. The garbler and the evaluator engage in an OT protocol where the garbler provides labels for inputs wires belonging to the evaluator and the garbler provides the bit-expression of $x_1$. The evaluator will receive the labels corresponding to $x_1$ without exposing $x_1$ to the garbler via OT. This step outputs the correct labels $\textsf{label}_{x_i}$ corresponding to actual input $x_i$.
    % first randomly generates the labels $\{\textsf{label}_{i,0},\textsf{label}_{i,1}\}_{i\in [n]}$ for inputs of $C$, with $n$ as the total number of input wires. Then it takes as input the security parameter, a boolean circuit $C$ (with $n$ input wires) and a set of labels $\{\textsf{label}_{i,0},\textsf{label}_{i,1}\}_{i\in [n]}$, it outputs a garbled circuit $\Tilde{C}$. Here $label_{i,b}$ represents assigning the value $b \in \{0,1\}$ to the $i^{\text{th}}$ input wire. 
    \item $\textsf{GC.Eval}(\Tilde{C},\{\textsf{label}_{x_i}\}_{ i \in \{1,2\}})\rightarrow y$: Used by the evaluator. The evaluator starts from the first gate and uses two input labels with the garbled table $\Tilde{C}_0$ to obtain the output. Then the evaluator performs the evaluation for each gate in order until obtaining the output $y$ of the last gate, where $y = C(x_1,x_2)$.
    % On input a garbled circuit $\Tilde{C}$ and $\{\textsf{label}_{i,x_i}\}_{i \in [n]}$ corresponding to an input $x\in \{0,1\}^{n}$, $\textsf{Eval}$ outputs a string $y=C(\boldmath{X})$
\end{itemize}
A two-party Garbled Circuit scheme can be easily extended to a multiparty computations scheme of arbitrary complexity with an arbitrary number of participants providing inputs. The input providers obtain the labels for their private inputs via OT from the garbler, then they transmit the labels to the evaluator, which reveals nothing about the inputs since the labels are random.

\subsubsection{Additive Secret Share}
Let $q$ be a prime. A 2-of-2 additive secret sharing of $x \in \mathbbm{Z}_{q}$ is a pair $(\langle x \rangle_{1},\langle x \rangle_{2})=(x-r,r)\in \mathbbm{Z}_{q}^{2}$ for a random $r \in \mathbbm{Z}_{q}$ such that $x=\langle x \rangle_{1}+\langle x \rangle_{2}$. Additive secret sharing is perfectly hidden, for example, given a share $\langle x \rangle_{1}$ or $\langle x \rangle_{2}$, the value $x$ is perfectly hidden.
 $n$-of-$n$ additive secret sharing can be easily extended from 2-of-2 additive secret sharing. A $n$-of-$n$ additive secret sharing of $x \in \mathbbm{Z}_{q}$ can be constructed by generating $(r_1,...,r_{n-1}) \in \mathbbm{Z}_{q}^{n-1}$ uniformly at random and setting $r_n = x-\sum_{i=1}^{n-1}r_i $(mod $q$).
 \subsection{Prior Work: DELPHI}
 Here we describe the protocol we built upon: DELPHI \cite{Delphi}. DELPHI is a cryptographic neural network inference protocol. There are two parties in the system of DELPHI: the user and the server. DELPHI uses additive secret sharing pre-built with homomorphic encryption to evaluate linear layers, and garbled circuits pre-built to evaluate activation functions. Let $\mathbf{F} = \{\mathbf{F}_i\}_{i=1}^{L}$ described in Section~\ref{section:splitnn} be the deployed model on the server, $\mathbf{x}_{i} = \mathbf{F}_{i-1}(\mathbf{F}_{i-2}(...(\mathbf{F}_{1}(\mathbf{x}_{1}))))$ is the the intermediate prediction result of the first $i-1$ layers where $\mathbf{x}_{1} \in \mathbbm{Z}_{q}$ is the input held by the user.
 % Figure~\ref{fig:delphi} shows the preprocessing and inference phase of DELPHI. 
 
 In the preprocessing phase, the client and the server pre-compute the secret share which will be used for the inference phase. For each linear layer $i\in[L]$, with the user providing the ciphertext of $\mathbf{r}_i$ and the server providing the plaintext of $\mathbf{F}_i$ and $\mathbf{s}_i$, two parties interact to compute $\mathbf{F}_i \mathbf{r}_i-\mathbf{s}_i$ where $\mathbf{r}_i$ and $\mathbf{s}_i$ are the random masks for building secret shares, resulting in each of the party holding one of the two secret shares of $\mathbf{F}_i \mathbf{r}_i$. For ReLUs following linear layers, the server constructs a garbled circuit $\Tilde{C}_{i}$ based on circuit $C$ computing ReLU function. Then the server transmits $\Tilde{C}_{i}$ to the user and input labels of $\mathbf{F}_i \mathbf{r}_i-\mathbf{s}_i$ and $\mathbf{r}_{i+1}$ via Oblivious Transfer to the user. We note that the algorithm for computing the ciphertext of $\mathbf{F}_i \mathbf{r}_i$ in DELPHI relies on GAZELLE's design \cite{Gazelle} for HE linear operations (convolutional and fully-connected layer). We do not provide details of this algorithm but represent it as a function \textsf{Lin-OP}. With function \textsf{Lin-OP}, ciphertext $ct_{\mathbf{r}_i}$ encrypting $\mathbf{r}_i$ and plaintext $pt_{F_i}$ encoding $\mathbf{F}_i$ as input, \textsf{HE.Eval} will output the ciphertext of $\mathbf{F}_i \mathbf{r}_i$.
 
 In the inference phase, assuming $\mathbf{x}_i$ is the prediction result of the first $i-1$ layers, at the beginning of the $i^{\text{th}}$ layer, the server and the user hold the shares of $\mathbf{x}_i$: $\mathbf{x}_{i}-\mathbf{r}_{i}$ and $\mathbf{r}_{i}$. The server evaluates to get $\mathbf{F}_i(\mathbf{x}_i-\mathbf{r}_i)+\mathbf{s}_i$, resulting in the server and the user holding the secret shares of $\mathbf{F}_i\mathbf{x}_i$. For evaluating ReLU, the server transmits the input labels corresponding to its secret share to the user. Then the client evaluates $\Tilde{C}_{i}$ to obtain a secret share of the ReLU output.

 \section{Secure Inference: Hierarchical Server Configurations}
 % \section{Inference with Model Splitting}
Our proposed setup involves splitting the neural network model across three server nodes. We assume that the user communicates with a gateway server that holds part of the neural network model. The remainder of the model is stored on remote servers that do not directly communicate with the user.  One motivation for considering such a scenario is that the interactions of the user and the server are required in computing each layer in the DELPHI protocol. This increases the computational and communication load on the user. Furthermore, the user gets access to information about the neural network including the shape of each layer and the total number of model layers since the user needs to generate randomness for each linear layer to build up additive secret shares, which is undesirable.  By deploying a subset of the model on remote servers that do not directly interact with the user we can partially address these concerns.
\subsection{Baseline Configuration}
We now explain why a straightforward extension of the DELPHI protocol that splits the neural network model across two server nodes as shown in Fig.~\ref{fig:splitnn} will lead to information leakage. In this setting, the user and the gateway server as well as the gateway server and the remote server respectively execute DELPHI to evaluate the two split models in order as Figure~\ref{fig:splitnn} shows. In this case, the user can be offline when two servers evaluate the second split model and the architecture of the second split model can be kept secret from the user. 

Let us consider deploying $\mathbf{F}_{\uppercase\expandafter{\romannumeral1}},\mathbf{F}_{\uppercase\expandafter{\romannumeral2}}$ (described in Section~\ref{section:splitnn}) on server $A$ and $B$ respectively. When evaluating $\mathbf{F}_{\uppercase\expandafter{\romannumeral2}}$ involving server $A$ and $B$, server $B$ holds the masked intermediate prediction result $\mathbf{x}_i-\mathbf{r}_i$ while randomness $\mathbf{r}_i$ is only contributed by server $A$. Note that if the two servers collude to recover $\mathbf{x}_i$ they can get access to full parameters of the first $i-1$ layers, then execute the white-box model inversion attack proposed in \cite{10.1145/3359789.3359824} to reconstruct the user's input.
\begin{figure}[h]
    \centering
    \includegraphics[scale=0.44]{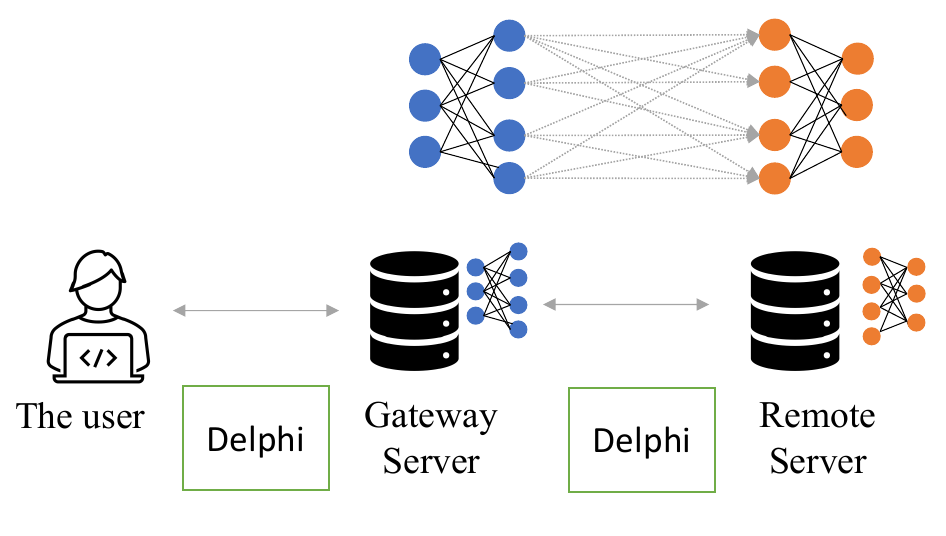} 
    \caption{DELPHI extended to Split Neural Network}
    \label{fig:splitnn}
\end{figure}
\subsection{Proposed Scheme}
We consider a setting with a hierarchy of servers where there are a gateway server and two remote servers as Figure~\ref{fig:arch} shows. We deploy $\mathbf{F}_{\uppercase\expandafter{\romannumeral1}}$ on the gateway server A and the processing of it follows DELPHI. For the second split model $\mathbf{F}_{\uppercase\expandafter{\romannumeral2}}$, the linear layers are randomly divided into 2 shares. We distribute the shares among 2 remote servers. 
% The preprocessing phase and the inference phase of $\mathbf{F}_{\uppercase\expandafter{\romannumeral1}}$ exactly follow DELPHI. Note in the preprocessing phase of DELPHI, the user and the server collaborate to compute two shares of $\mathbf{F}_i \mathbf{r}_i$ where $\mathbf{r}_i$ is the randomness generated by the user.
% Note the key idea of DELPHI is that in the preprocessing phase two parties use HE to build two secret shares of $\mathbf{F}_i\mathbf{r}_i$ where $\mathbf{r}_i$ is randomly generated by the user. In the inference phase, $\mathbf{r}_i$ will be used as the mask for the intermediate prediction result $\mathbf{x}_i$. The evaluation of the masked input for each layer can be processed in plaintext and the masking term can be canceled by the secret share built in the preprocessing phase. 
% However, direct extension from DELPHI to such setting 
% Then, we make use of Multiparty Homomorphic Encryption \cite{mouchet2021multiparty} to build the additive secret shares among the gateway server and remote servers.   Thus, the evaluation can only be conducted when all the servers are present. The distribution of model parameters will not expose the model parameters as long as all the servers do not collude.                                   \begin{figure}[h]
%     \centering
%     \includegraphics[scale=0.37]{arch.png} 
%     \caption{Protocol Architecture}
%     \label{fig:arch}
% \end{figure}
We summarize the main advantages of our proposed architecture as follows:  1) hide partial architecture information of the model from the user and reduce the user's computation and communication cost 2) defend the possible white-box model inversion attack resulting from the model leakage and intermediate prediction result recovery caused by server collusion.  We give the detailed description of the protocol in Section~\ref{section:protocol-description}.

\section{System Overview}
% \begin{figure}[h]
%     \centering
%     \includegraphics[scale=0.25]{Setting.png} 
%     \caption{System setting}
%     \label{fig:setting}
% \end{figure}
As shown in Figure~\ref{fig:arch}, there are four parties in our setting: a user, and hierarchical service providers including server A, B, and C. The hierarchy contains two layers of servers: gateway server A in the first layer, and server B and C in the second layer. Three servers cooperate to provide prediction service of a model $\mathbf{F}$. The model $\mathbf{F} $ can be expressed as $\mathbf{F}=(\mathbf{F}_{A}:\{\mathbf{F}_{i}\}_{i \in \{1,...,l\}},\mathbf{F}_{BC}:\{\mathbf{F}_{i}\}_{i\in \{l+1,...,L\}})$ where $L$ is the total number of model layers. Server A holds the first $l$ layers $\mathbf{F}_{A}$. For $i^{\text{th}} $ layer $i \in \{l+1,...,L\}$, Server B and C hold the random share of linear layer $\mathbf{F}_{i}$, which can be expressed as $\mathbf{F}_{i}=\sum_{j=2}^3\mathbf{F}_{i}^{j}$ ($j=2$ for server B, $j=3$ for server C).
\subsection{Threat Model and Privacy Goals}\label{section:privacy-goals}
\paragraph{Threat model.} {We assume a passive-adversary model with corruption of a user or up to two servers (any two). The parties will not deviate from the protocol but the adversary can corrupt the user or up to $2$ servers so can know corrupted parties' private data and the observations. We assume the adversary cannot corrupt the user and the server simultaneously. Additionally, we assume server A belongs to the service provider while server B and C are external. Deployment of random shares on external servers won't compromise the proprietary model privacy.}
% For the user, we assume it is honest but curious. The user will follow the protocol but try to infer the model parameters. We assume the user cannot collude with any server.

% \subsection{The Privacy Goals}\label{section:privacy-goals}
\paragraph{Privacy Goals.}
For the user, we aim to design a protocol that enables the user to only learn the result of the inference and the architecture of the model on gateway server A. All other information about the model including the model parameters, the total number of layers, and the architecture of the layers on server B, C should be kept secret from the user. 

Assuming up to 2 servers corrupted by the adversary, we have the privacy goals for the servers as the followings: 1) the adversary should not learn any honest party's model parameters or the intermediate prediction result $\mathbf{x}_i$; 2) the adversary should not infer the user's input.

\section{The Protocol: Formal Description}\label{section:protocol-description}
In this section, we introduce our secure inference protocol. The protocol takes as input the user's data $\mathbf{x}_{1}\in \mathbbm{Z}_{q}$ and the split model $(\mathbf{F}_{A}=\{\mathbf{F}_{i}\}_{i \in \{1,...l\}},\mathbf{F}_{BC}=\{\mathbf{F}_{i}\}_{i \in \{l+1,...L\}})$, and enables four parties to execute secure inference collaboratively. The protocol {can be considered as the combination of a 2-PC protocol, executed between the user and the gateway server for the inference of $\mathbf{F}_{A}$, and a 3-PC protocol, executed between the gateway server and two remote servers for the inference of $\mathbf{F}_{BC}$. The 2-PC protocol is adapted from DELPHI \cite{Delphi} and the 3-PC protocol is designed to achieve the secure neural network computations in the presence of dishonest majority of honest-but-curious servers}. The protocol consists of three phases as shown in Figure~\ref{fig:process}: a setup phase, a preprocessing phase, and an online inference phase. 

\begin{figure}[h]
    \centering
    \includegraphics[scale=0.52]{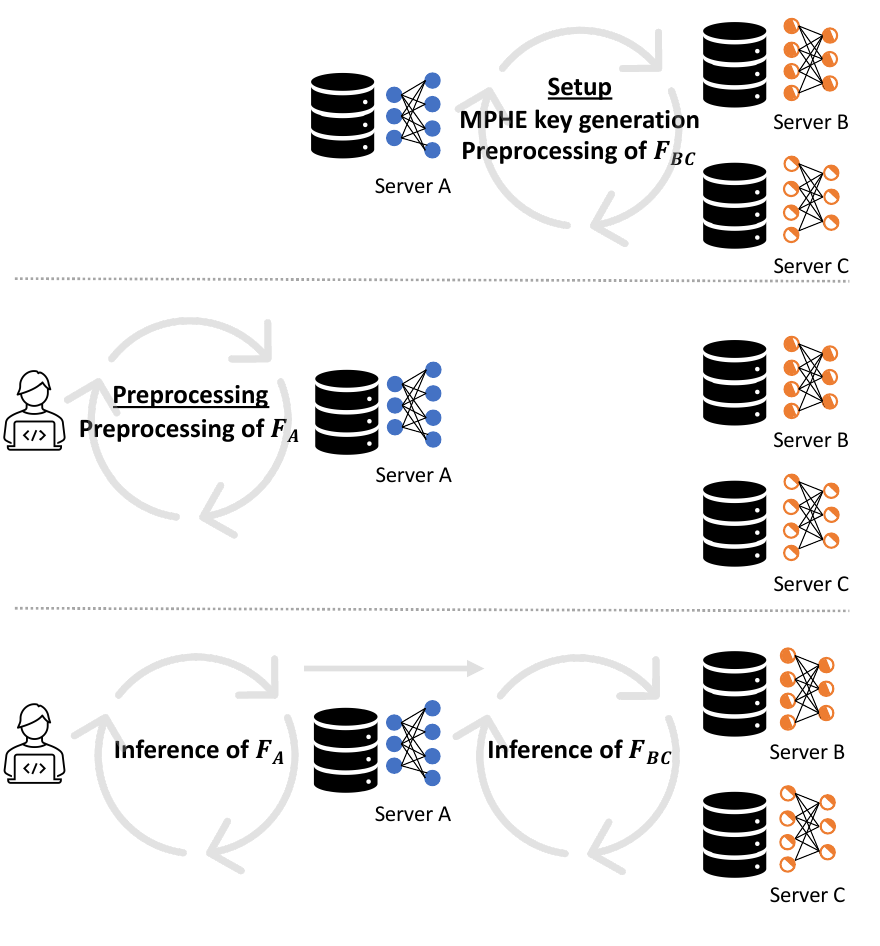} 
    \caption{SECO's 3 phases}
    \label{fig:process}
\end{figure}

\subsection{Setup Phase}
The setup phase is executed before the user appears in the system.
\subsubsection{Key Generation}
This step is to generate keys for homomorphic computation. 
\begin{enumerate}
    \item Each server runs \textsf{HE.KeyGen} to generate key pair $(\textsf{sk}_{j},\textsf{pk}_{j})$ ($j=0$ for the user, $j=1$ for server A, $j=2$ for server B, $j=3$ for server C).
    \item Server A collects the public keys from the three servers, runs $\textsf{MPHE.DKeyGen}$ to obtain a common public key $\textsf{cpk}$, then transmits $\textsf{cpk}$ to the user, server B and C. 
    % \item The user broadcasts $\textsf{pk}_{0}$ to each server in the system through Server A.
\end{enumerate}
At the end of the step, each server holds a common public key \textsf{cpk}.
\subsubsection{Model Preparation}: This step is the preparation for $\mathbf{F}_{BC}$ which can be executed before the user appears in the system. Like in DELPHI, the key insight is to pre-compute the additive secret shares of $\mathbf{F}_{i}\mathbf{r}_{i}$ for each linear layer of $\mathbf{F}_{BC}$. Server A, B, and C collaboratively compute the secret shares of $\mathbf{F}_{i}\mathbf{r}_{i}$ with MPHE scheme, where $\mathbf{F}_{i}$ is distributed among server B and C and the randomness $\mathbf{r}_{i}$ is solely contributed by the user if $i=l+1$ or jointly contributed by three servers if $i\in\{l+2,...,L\}$. The preprocessing of the $l+1^{\text{th}}$ linear layer is required for a secure transition from the inference of $\mathbf{F}_{A}$ to $\mathbf{F}_{BC}$. Additionally, three servers construct garbled circuits for each activation layer in $\mathbf{F}_{BC}$. 
\paragraph{Preprocessing for other linear layers.} For the $i^{\text{th}}$ layer $i \in \{l+2,...L\}$, server A, B, and C contribute to the randomness $\mathbf{r}_{i}$ and collaborate to compute three secret shares of $(\mathbf{F}^2_{l+1}+\mathbf{F}^3_{l+1})\mathbf{r}_{l+1}$ . With server B and C providing the ciphertexts of private model parameter $\mathbf{F}^j_i$, input randomness $\mathbf{r}^j_i$ and output randomness $\mathbf{s}^j_i$, server A computes its share $\mathbf{E}_{i}^1=(\mathbf{F}^2_i+\mathbf{F}^3_i)(\mathbf{r}^1_{i}+\mathbf{r}^2_{i}+\mathbf{r}^3_{i})-\mathbf{s}^2_{i}-\mathbf{s}^3_{i}$, resulting in building the additive shares of $(\mathbf{F}^2_i+\mathbf{F}^3_i)(\mathbf{r}^1_{i}+\mathbf{r}^2_{i}+\mathbf{r}^3_{i})$ among three servers. See Algorithm~\ref{Pre-com-lin-L} for details. In Algorithm~\ref{Pre-com-lin-L}, we use \textsf{MPHE.DisDec} (Algorithm~\ref{DistributeDec}) which is a protocol for the distribute decryption of MPHE scheme involving three servers.
\paragraph{Preprocessing for activation layers}
Three servers collaborate to build the garbled circuit on server A and prepare the evaluation of the activation layers. 
{With server B as the garbler, server C as the evaluator, and server A  providing input as the third party, server B first prepares the garbled circuit $\Tilde{C}_i$ and labels by running \textsf{GC.Garble}$(\textsf{Params}$,
$C_i)$ where \textsf{Params} is the security parameters and $C_i$ is described in Algorithm~\ref{alg:C_i_ABC}. To transmit the labels corresponding to actual input values, server B sends the labels of the input value provided by itself to server C through a public channel. For input values from server A and C, server A and C first obtain the actual labels from server B via OT. Then server C collects the labels obtained by server A through a public channel.}
\subsection{Preprocessing Phase}
In the preprocessing phase, the user first generates a key pair $(\textsf{sk}_{0},\textsf{pk}_{0})$ and broadcasts $\textsf{pk}_{0}$ to each server in the system through Server A. Server A sends common public key $\textsf{cpk}$ to the user. For preprocessing of $\mathbf{F}_{A}$, the user and server A collaboratively pre-compute the secret shares of weights for linear layers and construct garbled circuits for activation layers. 
\subsubsection{Preprocessing for $\mathbf{F}_{A}$}

% This step follows the approach in DELPHI as Algorithm~\ref{pre-com-lin-1} and Algorithm~\ref{pre-bui-GC-1} shows in Appendix~\ref{sec:peseudo}. 
This step follows the approach in DELPHI.
% as Algorithm~\ref{pre-com-lin-1} and Algorithm~\ref{pre-bui-GC-1} shows in Appendix~\ref{sec:peseudo}. 
% The user and server A collaboratively pre-compute the secret shares of weights for linear layers and construct garbled circuits for activation layers in $\mathbf{F}_{A}$.

\paragraph{Preprocessing for linear layers.}
For the $i^{\text{th}}$ linear layer $i \in [1,...,l]$, with the user providing the ciphertext of $\mathbf{r}_i$, server A computes the ciphertext of $\mathbf{F}_i \mathbf{r}_i-\mathbf{s}_i$ where $\mathbf{r}_i$ is the randomness of the layer's input size and $\mathbf{s}_i$ is the randomness of the layer's output size, then the user decrypts it using $\textsf{sk}_0$, resulting in the user holding $\mathbf{F}_i \mathbf{r}_i-\mathbf{s}_i$ and server A holding $\mathbf{s}_i$ which are two additive shares of $\mathbf{F}_i \mathbf{r}_i$.
% \begin{algorithm}
% \caption{Pre-Compute the share for the $i^{\text{th}}$ linear layer, $i \in \{1,...l\}$} 
% \label{pre-com-lin-1}
% {{
%     \begin{algorithmic}[1]
%     \Statex \textbf{Inputs}: $(q_i, w_{i})$: the input and output shape of the $i^{\text{th}}$ layer; $(\textsf{sk}_0, \textsf{pk}_0)$: the plain HE key pair generated by the user
%     \Statex \textbf{The user}: 
%     \Statex \hskip2.0em $\mathbf{r}_{i}\leftarrow \mathbbm{Z}_{q}^{q_i}$ 
%     \Statex \hskip2.0em Encrypt $ct_{r_i}=\textsf{HE.Enc}(\textsf{pk},\mathbf{r}_{i})$
%     \Statex \hskip2.0em Transmit $ct_{r_i}$ to server A 
%     \Statex \textbf{Server A}:
%     \Statex \hskip2.0em $\mathbf{s}_{i} \leftarrow \mathbbm{Z}_{q}^{w_i}$
%     % \Statex \hskip2.0em $\mathbf{d}_{i+1} \leftarrow \mathbbm{Z}_{q}^{q_i}$
%     \Statex \hskip2.0em Compute 
%     $\textsf{HE.Enc}(\textsf{pk},\mathbf{F}_{i}\mathbf{r}_{i}-\mathbf{s}_{i})$ and transmit this ciphertext to the user
%     \Statex \textbf{The user}:
%     \Statex \hskip2.0em Decrypt the ciphertext to obtain $\mathbf{F}_{i}\mathbf{r}_{i}-\mathbf{s}_{i}$ using $\textsf{sk}_0$
%     \end{algorithmic}}}
% \end{algorithm}

\paragraph{Preprocessing for activation layers.} 
 After the preprocessing for linear layers, the user and server A collaborate to build the garbled circuits for activation layers in $\mathbf{F}_A$. $C_i$ is the circuit to compute the activation function. Server A will execute \textsf{GC.Garble} to generate a set of the garbled tables $\Tilde{C}_i$ according to gates in circuit $C_i$ and labels corresponding to all input wires in $C_i$. Then server A transmits the labels for actual input values to the user. For input values held by server A such as one-time password $\mathbf{d}_{i+1}$, server A generates labels based on the bit expression of $\mathbf{d}_{i+1}$ and directly sends it to the user. For input values held by the user such as $\mathbf{F}_{i}\mathbf{r}_{i}-\mathbf{s}_{i},\mathbf{r}_{i+1}$, server A and the user engage in OT protocol to let the user obtain the corresponding labels without server A knowing the actual value.

\begin{algorithm}
\caption{$\textsf{MPHE.DisDec}$} 
\label{DistributeDec}
{\small{
    \begin{algorithmic}[1]
    \Statex \textbf{Inputs}: $ct$: the ciphertext of $m$; ${\textsf{sk}_j}$: the plain secret key generated by the contributor of $\textsf{cpk}$
    \Statex \textbf{Server A}: 
    \Statex \hskip2.0em Transmit $ct$ to server B and C
    \Statex \hskip2.0em $\textsf{MPHE.Reconstruct}(ct,\textsf{sk}_{1})\rightarrow pd_1$
    \Statex \textbf{Server B and C}:
    \Statex \hskip2.0em $\textsf{MPHE.Reconstruct}(ct,\textsf{sk}_{j}) \rightarrow pd_j$
    \Statex \hskip2.0em Transmit $pd_j$ to server A
    \Statex \textbf{Server A}:
    \Statex \hskip2.0em Run $\textsf{MPHE.Dec}(ct,\{pd_j\}_{\{j \in \{1,2,3\}\}})$ to obtain $m$
    \end{algorithmic}}}
\end{algorithm}
\paragraph{Preprocessing for the $l+1^{\text{th}}$ linear layer.} 
In this step, the protocol involves 4 parties to prepare for the transition from the inference of $\mathbf{F}_{A}$ to $\mathbf{F}_{BC}$. 
The user, server A, B, and C collaborate to build the additive shares of $(\mathbf{F}^2_{l+1}+\mathbf{F}^3_{l+1})\mathbf{r}_{l+1}$ among three servers, that is, with server B and C providing the ciphertexts of their own $\mathbf{F}^j_{l+1}$ and the ciphertxts of the output randomness $\mathbf{s}^j_{i+1}$, with the user providing the ciphertext of input randomness $\mathbf{r}_{l+1}$, server A homomorphically computes the ciphertext of $\mathbf{E}_{l+1}=(\mathbf{F}^2_{l+1}+\mathbf{F}^3_{l+1})\mathbf{r}_{l+1}-\mathbf{s}^2_{l+1}-\mathbf{s}^3_{l+1}$ and then involves server B and C to decrypt it. At the end of the preprocessing of the $l+1^{\text{th}}$ linear layer, three server holds the respective secret share of $(\mathbf{F}^2_{l+1}+\mathbf{F}^3_{l+1})\mathbf{r}_{l+1}$. Note all the encryption in this stage are conducted with MPHE public key \textsf{cpk}. 
% See Algorithm~\ref{Pre-com-lin-l-1} in Appendix~\ref{sec:peseudo} for details.
%  \begin{figure*}[htbp]
%     \centering
%     \includegraphics[scale=0.24]{seco_pre.png} 
%     \caption{SECO's preprocessing phase for linear layers of the second split model}
%     \label{fig:seco_pre}
% \end{figure*}
\begin{algorithm}
\caption{ Pre-Compute the share for the $i^{\text{th}}$ linear layer, $i \in \{l+2,...L\}$} 
\label{Pre-com-lin-L}
{\small{
    \begin{algorithmic}[1]
    \Statex \textbf{Inputs}: $(q_i, w_{i})$: the input and output shape of the $i^{\text{th}}$ layer; $\textsf{cpk}$: common public key
    % ; $j$: server index ($1$ for server A, $2$ for server B, $3$ for server C)
    \Statex \textbf{Server A}: 
    \Statex \hskip2.0em $\mathbf{r}_{i}^1\leftarrow \mathbbm{Z}_{q}^{q_i}$
    \Statex \hskip2.0em Compute $ct_{r^1_i}=\textsf{MPHE.Enc}(\textsf{cpk}$,$\mathbf{r}^1_{i})$
    \Statex \textbf{Server B and C}: 
    \Statex \hskip2.0em $\mathbf{r}_{i}^j\leftarrow \mathbbm{Z}_{q}^{q_i}$, $\mathbf{s}_{i}^j\leftarrow \mathbbm{Z}_{q}^{w_i}$ 
    \Statex \hskip2.0em Compute $ct_{r^j_i}=\textsf{MPHE.Enc}(\textsf{cpk}$,$\mathbf{r}^j_{i})$,
    $ct_{s^j_i}=\textsf{MPHE.Enc}(\textsf{cpk},\mathbf{s}^j_{i})$,
    $ct_{F^j_i}=\textsf{MPHE.Enc}(\textsf{cpk},\mathbf{F}^j_{i})$
    \Statex \hskip2.0em Transmit $(ct_{r^j_i},ct_{s^j_i},ct_{F^j_i})$ to server A 
    \Statex \textbf{Server A}:
    \Statex \hskip2.0em \textsf{MPHE.Eval}($ct_{F^2_{i}},ct_{F^3_{i}},\textsf{Add}$)$\rightarrow ct_{F_{i}}$
    \Statex \hskip2.0em \textsf{MPHE.Eval}($ct_{s^2_{i}},ct_{s^3_{i}},\textsf{Add}$)$\rightarrow ct_{s_{i}}$
    \Statex \hskip2.0em \textsf{MPHE.Eval}($ct_{r^1_{i}},ct_{r^2_{i}},ct_{r^3_{i}},\textsf{Add}$)$\rightarrow ct_{r_{i}}$
    \Statex \hskip2.0em \textsf{MPHE.Eval}($ct_{F_{i}},ct_{r_{i}},\textsf{Lin-OP}$)$\rightarrow ct_{F_{i}r_{i}}$
    \Statex \hskip2.0em\textsf{MPHE.Eval}($ct_{F_{i}r_{i}},ct_{s_{i}},\textsf{Sub})
    \rightarrow ct_{F_{i}r_{i}-s_{i}}$
    \Statex \hskip2.0em Transmit $ct_{F_{i}r_{i}-s_{i}}$ to server B and C
    \Statex \textbf{Server A, B, C}:
    \Statex \hskip2.0em Run $\textsf{MPHE.DisDec}(ct_{F_{i}r_{i}-s_{i}},\{\textsf{sk}_j\}_{j \in \{1,2,3\}})$
    \Statex \hskip2.0em Server A obtains $\mathbf{E}_i=(\mathbf{F}^2_i+\mathbf{F}^3_i)\mathbf{r}_{i}-\mathbf{s}^2_{i}-\mathbf{s}^3_{i}$ where $\mathbf{r}_{i}=\mathbf{r}^1_{i}+\mathbf{r}^2_{i}+\mathbf{r}^3_{i}$ 
    
    \end{algorithmic}}}
\end{algorithm}
% \begin{algorithm}
% \caption{Construct the Garbled Circuits for the $i^{\text{th}}$ activation layer $i \in \{l+1,...,L\}$} 
% \label{pre-bui-GC-l}
% {{
%     \begin{algorithmic}[1]
%     \Statex \textbf{Inputs}: $C_{i}$: the circuit in Algorithm~\ref{alg:C_i_ABC}; $\mathbf{E}_{i}^1$: the pre-computed additive share held by server A; $\mathbf{r}_{i+1}^2$: the randomness generated by server B 
%     \Statex \textbf{Server B}: 
%     \Statex \hskip2.0em \textsf{GC.Garble}$(\textsf{Params},C_i)\rightarrow \Tilde{C}_{i},\textsf{label}$
%     \Statex \hskip2.0em Transmit $\Tilde{C}_{i}$ to server C
%     \Statex \textbf{Server A and B}:
%     \Statex \hskip2.0em \textsf{GC.Transfer}($\textsf{label},\mathbf{E}_{i}^1)\rightarrow $
%     $\textsf{label}_{\mathbf{E}_{i}^1}$
%     \Statex \hskip2.0em Server A obtains $\textsf{label}_{\mathbf{E}_{i}^1}$ transmits it to server C
%     \Statex \textbf{Server B and C}:
%     \Statex \hskip2.0em \textsf{GC.Transfer}($\textsf{label},\mathbf{r}_{i+1}^3)\rightarrow \textsf{label}_{\mathbf{r}_{i+1}^3}$
%     \Statex \hskip2.0em Server C obtains $\textsf{label}_{\mathbf{r}_{i+1}^3}$
%     \end{algorithmic}}}
% \end{algorithm}

\subsection{Online Inference Phase}
This section introduces the inference phase in the order of model layers. First, as preamble, the user computes $\mathbf{x}_{1}-\mathbf{r}_{1}$ and sends it to server A.
% \subsubsection{Inference}
\paragraph{Inference of $\mathbf{F}_{A}$} The inference of $\mathbf{F}_{A}$ involves the participation of the user and server A. Before evaluating the $i^{\text{th}}, i\in [1,...,l]$ layer, server A holds the masked input $\mathbf{x}_{i}-\mathbf{r}_{i}$ and the user holds the input randomness $\mathbf{r}_{i+1}$ for next layer. Server A first computes $\mathbf{F}_{i}(\mathbf{x}_{i}-\mathbf{r}_{i})+\mathbf{s}_{i}$ to evaluate the linear part and sends the labels of it to the user. Collecting the labels of all the input wires in $\Tilde{C}_i$, the user executes \textsf{GC.Eval} to obtain a masked ReLU output as the masked input for the next layer.
\begin{algorithm}
\caption{A circuit $C_{i}$ that computes $ReLU$ function of the $i^{\text{th}}$ layer $i \in \{l+1,...,L\}$} 
\label{alg:C_i_ABC}
{\small{
    \begin{algorithmic}[1]
    \Statex Input from server A: $(\mathbf{F}_{i}^2+\mathbf{F}_{i}^3)\mathbf{r}_{i}-\mathbf{s}_{i}^2-\mathbf{s}_{i}^3,\mathbf{r}_{i+1}^1$
    \Statex Input from server B: $\mathbf{F}_{i}^2(\mathbf{x}_{i}-\mathbf{r}_{i})+\mathbf{s}_{i}^2, \mathbf{r}_{i+1}^2$
    \Statex Input from server C: $\mathbf{F}_{i}^3(\mathbf{x}_{i}-\mathbf{r}_{i})+\mathbf{s}_{i}^3,\mathbf{r}_{i+1}^3$
    \State Compute $(\mathbf{F}_{i}^2+\mathbf{F}_{i}^3)\mathbf{x}_{i}=(\mathbf{F}_{i}^2+\mathbf{F}_{i}^3)\mathbf{r}_{i}-\mathbf{s}_{i}^2-\mathbf{s}_{i}^3+\mathbf{F}_{i}^2(\mathbf{x}_{i}-\mathbf{r}_{i})+\mathbf{s}_{i}^2+\mathbf{F}_{i}^3(\mathbf{x}_{i}-\mathbf{r}_{i})+\mathbf{s}_{i}^3$
    \State Compute $\mathbf{x}_{i+1}=ReLU((\mathbf{F}_{i}^2+\mathbf{F}_{i}^3)\mathbf{x}_{i})$
    \State Output $\mathbf{x}_{i+1}-\mathbf{r}_{i+1}=\mathbf{x}_{i+1}-\mathbf{r}_{i+1}^1-\mathbf{r}_{i+1}^2-\mathbf{r}_{i+1}^3$
    \end{algorithmic}}}
\end{algorithm}

\paragraph{Transition} After the evaluation of $\mathbf{F}_A$, server A holds $\mathbf{x}_{l+1}-\mathbf{r}_{l+1}$. Server A forwards it to server B and C as a transition step from the inference of $\mathbf{F}_A$ to $\mathbf{F}_{BC}$.
\paragraph{Inference of $\mathbf{F}_{BC}$}
The inference of $\mathbf{F}_{BC}$ involves the participation of server A, B, and C. At the beginning of the inference for $i^{\text{th}}$ layer $i \in \{l+1,...,L\}$, server B and C hold $\mathbf{x}_{i}-\mathbf{r}_{i}$ where $\mathbf{r}_{i}$ is solely contributed by the user when $i=l+1$, and is contributed by server A, B and C for linear other layers of $\mathbf{F}_{BC}$. Server A, B and C collaborate
% run Algorithm~\ref{infer_L}
to do inference on $i^{\text{th}}$ layer $i \in \{l+1,...,L\}$. First server B and C will compute $\mathbf{E}_{i}^j=\mathbf{F}_{i}^j(\mathbf{x}_{i}-\mathbf{r}_{i})+\mathbf{s}_{i}^j$, resulting in three servers holding three additive shares of the real prediction result for the current layer $(\mathbf{F}_{i}^2+\mathbf{F}_{i}^3)\mathbf{x}_{i}$. The garbled circuit $\Tilde{C}_i$ (Algorithm~\ref{alg:C_i_ABC}) first aggregates shares to recover $(\mathbf{F}_{i}^2+\mathbf{F}_{i}^3)\mathbf{x}_{i}$, then computes ReLU of it, finally masks the ReLU output with the randomness of the next layer contributed by three servers. The entire process is in the encrypted domain. The garbled circuit for each ReLU layer of $\mathbf{F}_{BC}$ are located at server C. Thus server C can evaluate the garbled circuit when collecting all labels for input wires of $\Tilde{C}_i$ in a secure way. 

\begin{algorithm}
\caption{ Inference for the $i^{\text{th}}$ layer, $i \in \{l+1,...,L\}$}
\label{infer_L}
{\small{
    \begin{algorithmic}[1]
    \Statex \textbf{Inputs}:  $\mathbf{x}_{i}-\mathbf{r}_{i}$: held by server B and C; $\Tilde{C}_{i}$: the garbled circuit of circuit in Algorithm~\ref{alg:C_i_ABC}; \textsf{label}: random labels for input wires of Algorithm~\ref{alg:C_i_ABC}; $\mathbf{r}_{i+1}^1$: randomness from server A; 
    % \Statex \textbf{Server B and C}: 
    $\textsf{label}_{{E}_{i}^1},\textsf{label}_{r_{i+1}^2},\textsf{label}_{r_{i+1}^3}$: labels for actual input values held by server A
    \Statex \textbf{Server B and C}: 
    \Statex \hskip2.0em Compute $\mathbf{E}_{i}^j=\mathbf{F}_{i}^j(\mathbf{x}_{i}-\mathbf{r}_{i})+\mathbf{s}_{i}^j$
    % \Statex \textbf{Server A and C}: 
    % \Statex \hskip2.0em \textsf{GC.Transfer}(\textsf{label}, $\mathbf{E}_{i}^j$)$ 
    %   \rightarrow \textsf{label}_{\mathbf{E}_{i}^j}$
    % \Statex \hskip2.0em Server A obtains $\textsf{label}_{{F}_{i}^2({x}_{i}-{r}_{i})+{s}_{i}^2}$
    % \Statex \textbf{Server B and C}:
    \Statex \hskip2.0em \textsf{GC.Transfer}(\textsf{label}, $\mathbf{E}_{i}^3$)$ 
      \rightarrow \textsf{label}_{{E}_{i}^3}$
    \Statex \hskip2.0em Server C obtains $\textsf{label}_{{E}_{i}^3}$ 
    \Statex \textbf{Server C}: 
    \Statex \hskip2.0em \textsf{GC.Eval}$(\Tilde{C}_{i},\textsf{label}_{{E}_{i}^2},\textsf{label}_{{E}_{i}^3},\textsf{label}_{{E}_{i}^1},\textsf{label}_{r_{i+1}^1},$
    $\textsf{label}_{r_{i+1}^2})\rightarrow \mathbf{x}_{i+1}-\mathbf{r}_{i+1}^1-\mathbf{r}_{i+1}^2$
    \Statex \hskip2.0em Compute $\mathbf{x}_{i+1}-\mathbf{r}_{i+1}=\mathbf{x}_{i+1}-\mathbf{r}_{i+1}^1-\mathbf{r}_{i+1}^2-\mathbf{r}_{i+1}^3$
    \Statex \hskip2.0em Transmit $\mathbf{x}_{i+1}-\mathbf{r}_{i+1}$ to server B
    \end{algorithmic}}}
\end{algorithm}

\subsection{Output}
After the evaluation of $\mathbf{F}_{BC}$, server A holds the share $(\mathbf{F}_{L}^2+\mathbf{F}_{L}^3)\mathbf{r}_{L}-\mathbf{s}_{L}^2-\mathbf{s}_{L}^3$, server B holds $\mathbf{F}_{L}^2(\mathbf{x}_{L}-\mathbf{r}_{L})+\mathbf{s}_{L}^2$ and server C holds $\mathbf{F}_{L}^3(\mathbf{x}_{L}-\mathbf{r}_{L})+\mathbf{s}_{L}^3$.
\begin{enumerate}
    \item Server B and C computes  $\textsf{HE.Enc}(\textsf{pk}_{0},\mathbf{F}_{L}^i(\mathbf{x}_{L}-\mathbf{r}_{L})+\mathbf{s}_{L}^i)$ and sends this ciphertext to server A.
    \item Server A homomorphically add the received ciphertext to compute the ciphertexts $ct_{\hat{\mathbf{y}}}$ of prediction result $\hat{\mathbf{y}}$ where $\hat{\mathbf{y}}=({F}_{L}^2+{F}_{L}^3)\mathbf{x}_{L}$, then transmits $ct_{\hat{\mathbf{y}}}$ to the user.
    \item The user decrypts $ct_{\hat{\mathbf{y}}}$ with $sk_{0}$ to obtain the prediction result $\hat{\mathbf{y}}$.
\end{enumerate}
\subsection{Discussion}
{
We next highlight key technical contributions and the motivations for our design choices.

\paragraph{Secret Sharing for 3PC and Integration with MPHE}
To sum up, during the executions, the protocol processes model $\mathbf{F}_{A}$ and $\mathbf{F}_{BC}$. The processing of $\mathbf{F}_{A}$ follows DELPHI which uses plain HE scheme to build two secret shares of $\mathbf{F}_i\mathbf{r}_i$ among two parties in the preprocessing phase, where $\mathbf{r}_i$ is randomly generated by the user; In the inference phase, $\mathbf{r}_i$ will be used as the mask for the intermediate prediction result $\mathbf{x}_i$ and it can be canceled by the secret share built in the preprocessing phase. Our protocol follows this idea to build three secret shares of $\{\mathbf{F}_i\mathbf{r}_i\}_{i=l+1}^L$ in the preprocessing phase of $\mathbf{F}_{BC}$, where each of the servers contributes its own part to the randomness $\mathbf{r}_i$. However, using a plain HE scheme to compute the secret shares is not privacy-preserving in this case because the model parameters and randomness are distributed. To ensure that each of the three servers exclusively accesses its designated information and no extraneous data, we design a new secret sharing scheme as in Table~\ref{tab:secret-share}, so the shares can be computed with MPHE using the supported operations as described in Section~\ref{plainBFV}, enabling secure computation of linear functions.

\begin{table}
    \caption{The columns depict the three distinct segments of the deployed model, each illustrating contributions from the involved parties and different computational goals. }
{\footnotesize
    \centering
    \begin{tabular}{cccccc}
    \toprule
                        & $\mathbf{F}_{A}$ & $\mathbf{F}_{BC}$ & Transition \\ \midrule
User                    &  $\mathbf{r}_i$  &    &   $\mathbf{r}_i$         \\ \midrule
Server A                &  $\mathbf{F}_i$,$\mathbf{s}_i$  &  $\mathbf{r}_i^1$  &            \\ \midrule
Server B                &    &  $\mathbf{r}_i^2$,$\mathbf{F}_i^2$,$\mathbf{s}_i^2$ &   $\mathbf{F}_i^2$,$\mathbf{s}_i^2$         \\ \midrule
Server C                &    &  $\mathbf{r}_i^3$,$\mathbf{F}_i^3$,$\mathbf{s}_i^3$  &           $\mathbf{F}_i^3$,$\mathbf{s}_i^3$ \\ \midrule
 Shares & \multirow{2}{*}{$\mathbf{F}_i\mathbf{r}_i-\mathbf{s}_i$ } & {$(\mathbf{F}_i^2+\mathbf{F}_i^3)(\mathbf{r}_i^1+\mathbf{r}_i^2+\mathbf{r}_i^3)$} & {$(\mathbf{F}_i^2+\mathbf{F}_i^3)\mathbf{r}_i$} \\
to compute   &             &        $-(\mathbf{s}_i^2+\mathbf{s}_i^3)$          &  $-(\mathbf{s}_i^2+\mathbf{s}_i^3)  $     \\ \midrule
         
    \end{tabular}

    \label{tab:secret-share}}
\end{table}

\paragraph{Assignments of Servers in GC}
For the secure computation of non-linear functions, we deviate from the conventional approach of assigning the layer's input provider as the evaluator in Garbled Circuits (GC), like in DELPHI. Instead, we appoint a remote server as the evaluator and another as the garbler. This approach substantially diminishes the communication complexity during the inference phase, as communication now occurs exclusively between these two remote servers. Additionally, our protocol maintains the security of user's input even in the event of two remote servers being compromised, because of the gateway server's contribution to the masking term during the setup phase. Last but not the least, we design a transition from user-server interaction to server-server interaction to ensure no information leakage in this process.}

\section{Information Leakage Analysis}\label{sec:info-leak-ana}
In this section,  we show the protocol described in Section~\ref{section:protocol-description} achieves the privacy goals described in Section~\ref{section:privacy-goals} under a passive-adversary model. Based on Composition Theorem for semi-honest model (Theorem 7.3.3 in \cite{articleFoC}), by proving the security of the protocol in each phase, the security of the entire protocol is proven. 
% \begin{theorem} \label{theo:composition}
% Suppose that $g$ is
% privately reducible to $f$ and that there exists a protocol for privately computing $f$.
% Then there exists a protocol for privately computing $g$.
% \end{theorem}

We prove the security of the preprocessing and inference phase following the simulation paradigm \cite{lindell2017simulate}. According to the simulation paradigm, a protocol is secure if whatever can be computed by a party involved in the protocol can be computed given only the input and the output of the party. So the security can be modeled by defining a simulator that can generate the view of the party during the procedure given the input and the output. If the view of the party can be simulated based on the input and the output only and is computationally indistinguishable from the real view of the party, it implies the party only learns what can be computed given the input and the output, so the protocol is secure. 

\subsection{Setup phase}
The setup phase is independent of the rest of the protocol. For a given user and the servers, it has to be run only once. The protocol used in the setup phase is a composition of $\textsf{HE.KeyGen}(.)$ and $\textsf{MPHE.DKeyGen}(.)$. 

In Mouchet et al. \cite{mouchet2021multiparty}, it shows that the protocol in the setup phase can securely and correctly generate a valid plain BFV key pair. Provided with valid plain BFV key pair, \textsf{HE.Enc} and \textsf{MPHE.Enc} can output valid BFV ciphertext,
% also based on decisional-RLWE problem 
which guarantees the semantic security of HE and MPHE schemes described in Proposition~\ref{prop-he-security} and Proposition~\ref{prop-mphe-security}. The semantic security of HE and MPHE schemes provides the security basis for the following two phases.

\begin{prop}[HE semantic security]\label{prop-he-security}
    For any two messages $m_1, m_2$, no adversary has an advantage (better than $1/2$ chance) in distinguishing between distributions $\textsf{HE.Enc}(\textsf{pk}, m_1)$ and $\textsf{HE.Enc}(\textsf{pk}, m_2)$.
    % This requires the randomized $\textsf{pk}$ and the randomized encryption algorithm.
\end{prop}
\begin{prop}[MPHE-based MPC semantic security]\label{prop-mphe-security}
    For any subsets of at most colluded $N-1$ clients, for any two messages $m_1$ and $m_2$, no adversary has an advantage (better than $1/2$ chance) in distinguishing between distributions $\textsf{MPHE.Enc}(\textsf{cpk}, m_1)$ and \textsf{MPHE.Enc}(\textsf{cpk},$m_2$).
\end{prop}

\subsection{Preprocessing and Inference Phase}
In this section, we prove the privacy goals in Section~\ref{section:privacy-goals} are achieved in the preprocessing and inference phase following the simulation paradigm \cite{lindell2017simulate}.
% According to simulation paradigm, a protocol is secure if whatever can be computed by a party involved in the protocol can be computed given only the input and the output of the party. So the security can be modeled by defining a simulator that can generate the view of the party during the procedure given the input and the output. If the view of the party can be simulated based on the input and the output only and is computationally indistinguishable from the real view of the party, it implies the party learns only what can be computed given the input and the output, so the protocol is secure. 

% Provided with the security properties (Appendix~\ref{sec:securityProperty}) of the cryptographic blocks (section~\ref{sec:cryptographic-blocks}),

 We can prove the security of the preprocessing and inference phase by 1) providing the real view of the adversary for the case when the user or two servers are corrupted 2) describing simulators for cases where two servers or the user are corrupted 3) comparing each term in the view of the adversary and proving the indistinguishability. We have three privacy goals that are to preserve the privacy of 1) the user's input 2) the intermediate prediction result 3) the complete model parameters. Intuitively, by distributing the model parameters, the honest server's partial model parameters are protected thus the adversary cannot obtain the complete model. The user's input and the intermediate prediction result are protected by the randomness. To prove they are fully secure, we can determine the indistinguishability of each element in the view separately. Note Composition Theorem \cite{articleFoC} can be applied to the sequential compositions of arbitrary protocols involving multiple parties. Based on Composition Theorem \cite{articleFoC}, the security of the protocols for two phases is proven when the composing sub-protocols are secure. The security proof of cryptographic sub-protocols (section~\ref{sec:cryptographic-blocks}) are provided in \cite{mouchet2021multiparty}\cite{article}\cite{eprint-2005-12523} assuming the passive adversary model.
% We consider $\textsf{pk}_0,\textsf{cpk}$ as publicly known encryption key, $\{C_i\}_{i=1}^l, \{\mathbf{F}_i\}_{i=1}^l$ and $\$ as private inputs for server A,  

We name the protocols of two phases as \textsf{PRE} and \textsf{INF} and define the corresponding ideal functionality as 
% \begin{eqnarray}
\begin{align}
    f_{pre}&(\textsf{pk}_{0},\textsf{cpk},\{C_i\}_{i=1}^L,\{\mathbf{F}_i\}_{i=1}^l,\{\mathbf{F}_i^j\}_{i=l+1,j=2,3}^{i=L}) \nonumber
    \\&= \{\{\textsf{share}_{i}^U\}_{i=1}^{l},\{\textsf{share}_{i}^A\}_{i=l+1}^{L},\{\Tilde{C}_i\}_{i=1}^L\} \\
    f_{inf}&(\mathbf{x}_1,\{\textsf{share}_{i}^U\}_{i=1}^{l},\{\textsf{share}_{i}^A\}_{i=l+1}^{L},\{\Tilde{C}_i\}_{i=1}^L) \nonumber
    \\&= \{\{\textsf{msk}_i\}_{i=1}^{L}\}
\end{align}
where $\{\textsf{share}_{i}^U=\mathbf{F}_i\mathbf{r}_i-\mathbf{s}_i\}_{i=1}^l, \{\textsf{share}_{i}^A=(\mathbf{F}_i^2+\mathbf{F}_i^3)(\mathbf{r}_i^1+\mathbf{r}_i^2+\mathbf{r}_i^3)-\mathbf{s}_i^2-\mathbf{s}_i^3\}_{i=l+1}^L, \{\textsf{msk}_{i}= \mathbf{x}_i-\mathbf{r}_i\}_{i=1}^L$. Next we we discuss the cases where server A and B are corrupted and the user is corrupted. 
% The proof for other cases where server B and C are corrupted or server A and C are corrupted can be easily extended from the following proof.
The proofs for other cases, wherein either servers B and C or servers A and C are corrupted, can be readily extended based on the methodology outlined in the following proof.
% in Appendix~\ref{sec:sec-proof}.

 \paragraph{Server A, B corrupted}
The real view $v_{AB}$ of the adversary corrupting server A and B includes:
\begin{enumerate}
    \item Ciphertexts set $ct_{S_{AB}} =\{\{ct_{r_i}\}_{i=1}^{l+1},\{ct_{r_i^2},ct_{r_i^3}\}_{i=l+2}^{L}, \\
    \{ct_{F_i^2},ct_{F_i^3},ct_{s_i^2},ct_{s_i^2}\}_{i=l+2}^{L},
    \{ct_{\textsf{share}_i^A}\}_{i=l+1}^L\}$
    \item The masked input  $\{\textsf{msk}_{i}=\mathbf{x}_i-\mathbf{r}_{i}\}_{i=1}^{l+1},\{\textsf{msk}_{i}=\mathbf{x}_i-\mathbf{r}_{i}^3\}_{l+2}^{L+1}$
     \item The secret shares set $\textsf{share}^{A}_{i}=(\mathbf{F}_{i}^2+\mathbf{F}_{i}^3)\mathbf{r}_{i}-\mathbf{s}_{i}^2-\mathbf{s}_{i}^3,i\in \{l+1,...,L\}$
    \item Garbled circuit $\{\Tilde{C}_i\}_{i=1}^L$ and labels $\{\textsf{label}_{i}\}_{i=l+1}^{L}$ 
\end{enumerate}
We define a simulator $S_{AB}$ that simulates the view of an adversary corrupting server A and B. For a given value $x$, we denote the simulated equivalent by $\Tilde{x}$ and the encryption of it by $ct_x$. $S_{AB}$ is given $\textsf{pk}_0,\textsf{cpk},\mathbf{F}_{A},\{\mathbf{F}_{i}^2\}_{i=l+1}^{L} , {(q_i, w_i)}_{i=1}^{L}$ and proceeds as follows:
\begin{enumerate}
    \item $S_{AB}$ chooses a uniform random tape for server A and B.
    \item In the preprocessing phase:
    \begin{enumerate}
    \item $S_{AB}$, as the user, computes $\{\Tilde{ct}_{r_i}=\textsf{HE.Enc}(\textsf{pk}_0, \mathbf{\Tilde{r}}_{i})\}_{i=1}^{l+1}$ and sends it to server A where $\mathbf{\Tilde{r}}_{i} \leftarrow \mathbbm{Z}_{q}^{q_i}$. Then $S_{AB}$ receives the evaluated ciphertext from server A.
    \item With $S_{AB}$ as the evaluator of the garbled circuits for $i^{\text{th}}, i\in \{1,...,l\}$ non-linear layer, $S_{AB}$ receives ${\Tilde{C}}_i$ from server A and labels corresponding to random input generated by $S_{AB}$.
    \item Server B computes and sends the ciphertexts 
    \begin{align}
        {ct}_{F_i^2}&=\textsf{MPHE.Enc}(\textsf{cpk},\mathbf{F}_i^2)\nonumber \\ {ct}_{r_i^2}&=\textsf{MPHE.Enc}(\textsf{cpk}, \mathbf{r}_{i}^2)\nonumber\\{ct}_{s_i^2}&=\textsf{MPHE.Enc}(\textsf{cpk}, \mathbf{s}_{i}^2) \nonumber
    \end{align}
    to server A.
    \item $S_{AB}$, as server C, sets $\mathbf{\Tilde{F}}_i^3=\mathbf{0}$, then generates $\Tilde{\mathbf{r}}_{i}^3\leftarrow \mathbbm{Z}_{q}^{q_i}, \Tilde{\mathbf{s}}_{i}^3\leftarrow \mathbbm{Z}_{q}^{w_i}$ and encrypts them with $\textsf{cpk}$; During distribute decryption, server B receives the evaluated ciphertext from server A and server A computes $\widetilde{\textsf{share}}_i^A = (\mathbf{F}_i^2+\mathbf{\Tilde{F}}_i^3)(\mathbf{r}_{i}^1+\mathbf{r}_{i}^2+\Tilde{\mathbf{r}}_{i}^3)-\mathbf{s}_{i}^2-\Tilde{\mathbf{s}}_{i}^3$ 
    \end{enumerate}
    \item In the inference phase:
    \begin{enumerate}
        \item Preamble: $S_{AB}$, as the user, sends $-\mathbf{\Tilde{r}}_1$ to server A to evaluate the first linear layer. Then $S_{AB}$ receives the corresponding GC label.
        \item $S_{AB}$ evaluates ${\Tilde{C}}_i$ to obtain the output of garbled circuits and sends it to server A.
        \item Transition: $S_{AB}$, as server C, receives the output of first $l$ layers
        \item $S_{AB}$ obtains labels corresponding to  $\Tilde{\mathbf{s}}_{i}^3$ and $\Tilde{\mathbf{r}}_{i+1}^3$ from server B, computes the output of garbled circuits and sends it to server B
    \end{enumerate}
\end{enumerate}
\begin{theorem} \label{theo:server-a}
The view simulated by $S_{AB}$ is computaionally indistinguishable from the real view $v_{AB}$ when the adversary is corrupting server A and B.
\end{theorem}
\begin{proof} We observe that the \textsf{PRE} and \textsf{INF} protocol can be privately reducible to the \textsf{HE}, \textsf{MPHE}, \textsf{GC} and other protocols. Therefore, the secuirty of \textsf{PRE} and \textsf{INF} protocol follow from the standalone security of each protocol by applying Composition Theorem for semi-honest model \cite{articleFoC}. The indistinguishability of ciphertext is guaranteed by the semantic security (Proposition~\ref{prop-he-security} and Proposition~\ref{prop-mphe-security}) of protocol \textsf{HE} and \textsf{MPHE}, and the indistinguishability of garbled circuits with labels is guaranteed by the security property of protocol \textsf{GC} \cite{article}. We only need to show the indistinguishability of the masked input and the the secret shares, then we can prove that \textsf{PRE} and \textsf{INF} preserve their security. Due to using different random values which share identical distribution hence are statistically equivalent, we get $\{\widetilde{\textsf{msk}}_{i}\}_{i=1}^{L}\overset{c}{\equiv} \{\textsf{msk}_{i}\}_{i=1}^{L},\{\widetilde{\textsf{share}}^{A}_{i}\}_{i=l+1}^L  \overset{c}{\equiv} \{\textsf{share}^{A}_{i}\}_{i=l+1}^{L}$. 
% By comparing the real view $v_{AB}$ with the view simulated by $S_{AB}$, we get
% \begin{enumerate}
%     % \item In the simulated world, server A uses the actual model parameter $\mathbf{M}_{A}$, real view 1. are indistinguishable from the simulated ones.
%     \item The relations $\Tilde{ct}_{S_{AB}} \overset{c}{\equiv} ct_{S_{AB}}$  follows the semantic security of BFV Scheme (Prop~\ref{prop-he-security} and Prop~\ref{prop-mphe-security}).
%     \item Following the security of the GC \cite{article}, $\{{\Tilde{C}}_i,\widetilde{\textsf{label}}_{i}\}_{i=1}^{L}\overset{c}{\equiv}\{\Tilde{C}_i,\textsf{label}_{i}\}_{i=1}^{L}$
%     \item Due to using different random values which share identical distribution hence are statistically equivalent, we get $\{\widetilde{\textsf{msk}}_{\mathbf{x}_i}\}_{i=1}^{L}\overset{c}{\equiv} \{\textsf{msk}_{\mathbf{x}_i}\}_{i=1}^{L},\{\widetilde{\textsf{share}}^{A}_{i}\}_{i=l+1}^L  \overset{c}{\equiv} \{\textsf{share}^{A}_{i}\}_{i=l+1}^{L}$
% \end{enumerate}
% Because of the independence of each term, the joint simulated view is computationally indistinguishable from the joint real view.
\end{proof}
The arguments above yields the following security property. 
\begin{prop}
    SECO introduced in Section~\ref{section:protocol-description} securely realizes $f_{pre}$ and $f_{inf}$ in the presence of semi-honest adversary controlling server A and B.
\end{prop}

\subsubsection{The user corrupted}
We prove that the simulated view and the real view are computationally indistinguishable, thus information such as the model parameters which is beyond the view cannot be computed if the adversary corrupts the user.

The real view $v_U$ of the adversary corrupting the user includes:
\begin{enumerate}
\item The masked input $\{\textsf{msk}_{i}=\mathbf{x}_i-\mathbf{r}_{i}-\mathbf{d}_{i}\}_{i=1}^{l}$
\item The additive share $\{\textsf{share}_i^U=\mathbf{F}_i\mathbf{r}_i-\mathbf{s}_i\}_{i=1}^{l}$
\item Garbled circuit $\{\Tilde{C}_i\}_{i=1}^l$ and labels $\{\textsf{label}_{i}\}_{i=1}^{l}$ 
\end{enumerate}

We define a simulator $S_{U}$ that simulates the view of an adversary corrupting the user. $S_{U}$ is given $\mathbf{x}_1, \textsf{pk}_0,\textsf{cpk}$ and proceeds as follows:
\begin{enumerate}
    \item $S_{U}$ chooses a uniform random tape for the user.
    \item In the preprocessing phase:
    \begin{enumerate}
    \item $S_{U}$, as server A, receives $\{ct_{r_i}=\textsf{HE.Enc}(\textsf{pk}_0, \mathbf{{r}}_{i})\}_{i=1}^{l+1}$ with $\mathbf{{r}}_{i}\leftarrow\mathbbm{Z}_{q}^{q_i}$. With $\mathbf{\Tilde{F}}_i$ set to $\mathbf{0}$, the user receives the evaluated $\{ct_{s_i}=\textsf{HE.Enc}(\textsf{pk}_0, -\mathbf{\Tilde{s}}_{i})\}_{i=1}^{l+1}$ from $S_U$ for a randomly chosen $\mathbf{\Tilde{s}}_{i}$ from $\mathbbm{Z}_{q}^{w_i}$.
    \item With $S_{U}$ acting as the garbler of the garbled circuits for $i^{\text{th}}, i\in \{1,...,l\}$ non-linear layer, the user receives ${\Tilde{C}}_i$ from $S_U$ and labels corresponding to $\mathbf{{r}}_{i}$ and $\mathbf{\Tilde{s}}_i$. $S_{U}$ sets the output of the circuit to be a random value.
 
    \end{enumerate}
    \item In the inference phase:
    \begin{enumerate}
        \item Preamble: $S_{U}$, as server A, receives $\mathbf{x}_1-\mathbf{{r}}_1$ from the user
        \item $S_{U}$, as server A, sets the  $\mathbf{F}_i=0$ evaluates to obtain $\Tilde{\mathbf{{s}}}_i$ and sends the corresponding simulated labels to the user.
        \item The user evaluates ${\Tilde{C}}_i$ to obtain the output of garbled circuits which is set randomly, then sends it to $S_{U}$.
    \end{enumerate}
\end{enumerate}
\begin{theorem} 
The view simulated by $S_{U}$ is computaionally indistinguishable from the real view $v_{U}$ when the adversary is corrupting the user.
\end{theorem}
\begin{proof}
Similar to the proof of Theorem~\ref{theo:server-a}, the views are computationally indistinguishable based on the security properties of \textsf{HE}, \textsf{MPHE}, \textsf{GC} and the use of randomness that are statistically equivalent.
\end{proof}
The arguments above yields the following security property. 
\begin{prop}
    SECO introduced in Section~\ref{section:protocol-description} securely realizes $f_{pre}$ and $f_{inf}$ in the presence of semi-honest adversary controlling the user.
\end{prop}
\section{Experiments}
In this section, we experimentally evaluate the performance of SECO and present our experimental results comparing SECO and other protocols. Currently, SECO supports the deployment of CNN models. 
\subsection{Implementation}
Our implementation is based on SEAL library \cite{sealcrypto} and DELPHI's open source code \footnote{https://github.com/mc2-project/delphi}. We first implemented the operations for multiparty key-generation and distribute-decryption for BFV scheme \cite{mouchet2021multiparty} in SEAL. Then we modified DELPHI's implementation to support the collaborative inference involving multiple parties as described in Section~\ref{section:protocol-description}. { For convolutional functions, we employ modulus-switching to reduce the ciphertext size as in Cryptflow2 \cite{Rathee_2020}}.
% , including extending encrypted neural network operations based on HE to be based on MPHE and extending garbled circuit scheme to a 3-party setting. 
We use the same BFV security parameters as DELPHI uses, with 8192 as the degree of polynomial modulus and 2061584302081 as plaintext modulus.
\begin{figure}[h]
\flushleft  
\subfigure[ResNet32 online time]{
\label{fig:online-time}
\includegraphics[width=8.5cm]{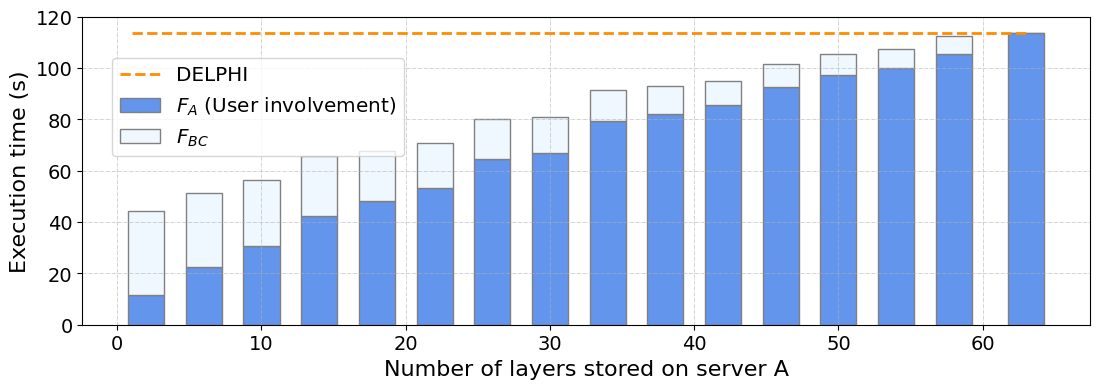}}
\subfigure[ResNet32 preprocessing time]{
\label{fig:pre-time}
\includegraphics[width=8.5cm]{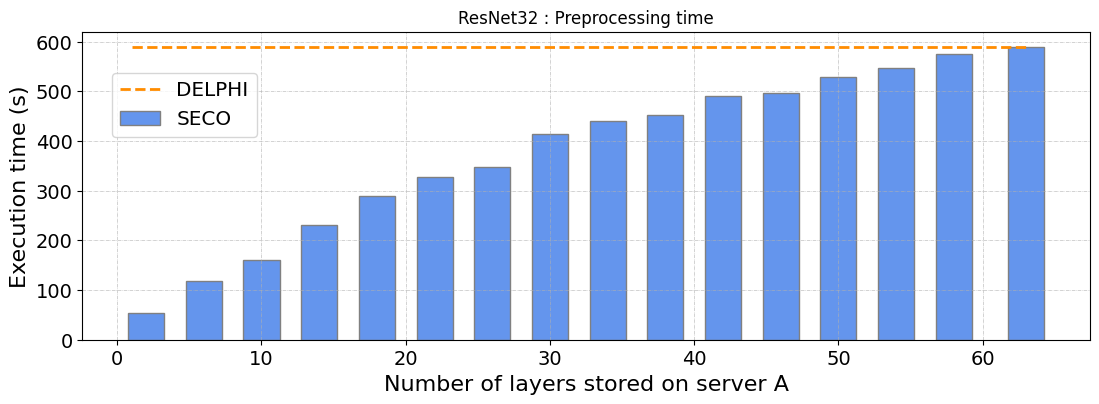}}
\caption{Comparison of SECO and DELPHI on execution time of ResNet32}
\label{fig:time}
\end{figure}

\begin{figure}[h]
\flushleft  
\subfigure[ResNet32 online communication]{
\label{fig:online-comm}
\includegraphics[width=8.5cm]{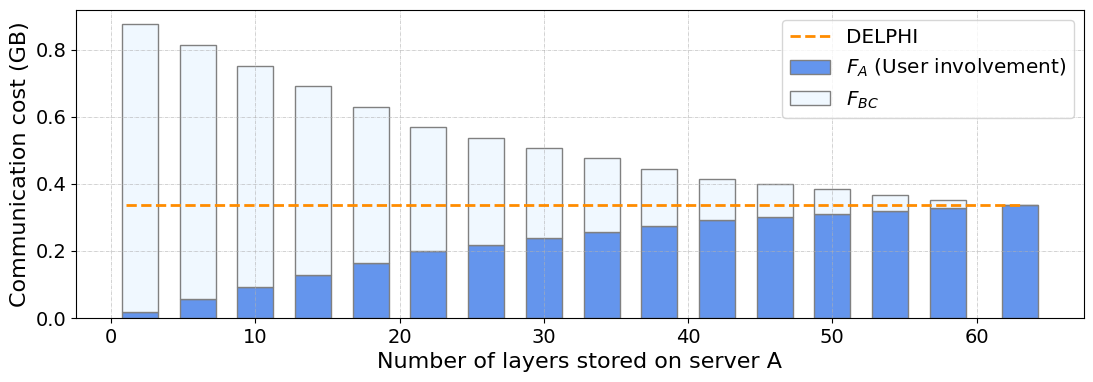}}
\subfigure[ResNet32 preproccessing communication]{
\label{fig:pre-comm}
\includegraphics[width=8.5cm]{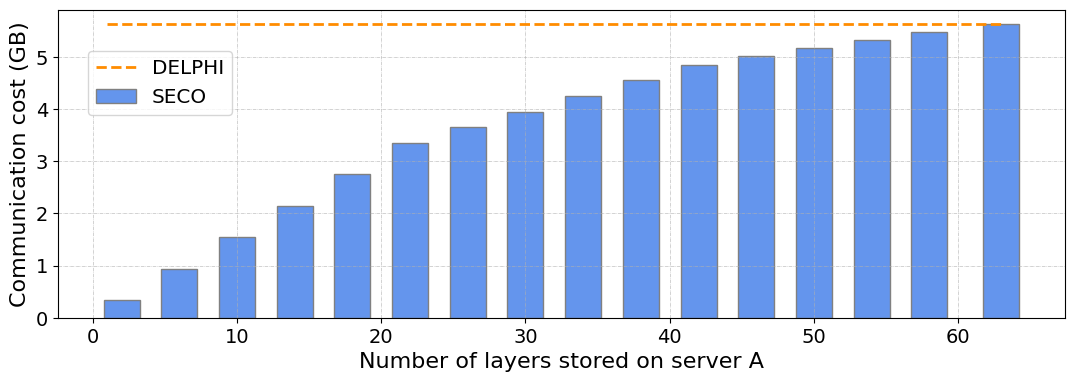}}
\caption{Comparison of SECO and DELPHI on communication cost of ResNet32}
\label{fig:communication}
\end{figure}

% \begin{figure}
% \centering  
% \subfigure[ResNet32 online communication]{
% \label{fig:online-comm}
% \includegraphics[width=7.8cm]{compare1.png}}
% \subfigure[ResNet32 Preproccessing communication]{
% \label{fig:pre-comm}
% \includegraphics[width=7.8cm]{compare2.png}}
% \caption{Comparison of SECO and DELPHI on communication cost of ResNet32}
% \label{fig:communication}
% \end{figure}

\subsection{Setup}
{We run our benchmarks in the WLAN setting with four instances in Australia or North America.}
All the instances have Intel Xeon Processor running at 2.3 GHz with 32GB of RAM. Each experiment is conducted 5 times, and the results presented are the average of these runs.
% \begin{table*}[]
% \centering
% \caption{Execution time and Communication cost of Single ReLU Operation}\label{tab:seco-perf-relu}
% {{
% \begin{tabular}{ccccc}
% \toprule
%      & \multicolumn{2}{c}{Execution Time (ms)} & \multicolumn{2}{c}{Communication Cost (byte)} \\
%      & Preprocessing       & Inference         & Preprocessing           & Inference            \\\midrule
% ReLU & 0.057343±0.00607    & 0.15066±0.00119   & 116442                  & 4378         \\
% \bottomrule 
% \end{tabular}}}
% \end{table*}

% \begin{figure}
% \centering  
% \subfigure[MiniONN preprocessing communication]{
% \label{fig:minionn-pre-com}
% \includegraphics[width=7.8cm]{minionn_pre_com.pdf}}
% \subfigure[MiniONN inference communication]{
% \label{fig:minionn-inf-com}
% \includegraphics[width=7.8cm]{minionn_inf_com.pdf}}
% \caption{Comparison of SECO and DELPHI on communication cost of MiniONN}
% \label{fig:minionn-com}
% \end{figure}

% \begin{figure}
% \centering  
% \subfigure[ResNet32 preprocessing communication]{
% \label{fig:resnet32-pre-com}
% \includegraphics[width=7.8cm]{resnet32_pre_com.pdf}}
% \subfigure[ResNet32 inference communication]{
% \label{fig:resnet32-inf-com}
% \includegraphics[width=7.8cm]{resnet32_inf_com.pdf}}
% \caption{Comparison of SECO and DELPHI on communication cost of ResNet32}
% \label{fig:resnet32-com}
% \end{figure}
\subsection{Performance}
\subsubsection{Compare with HE-MPC-based Protocols}
% {\color{blue}
We experimentally evaluate SECO in terms of execution time and communication cost for preprocessing and inference phase with the deployed model split at different layers. We present the comparison results of SECO, DELPHI \cite{Delphi}, and DELPHI-3. DELPHI-3 is a straightforward extension of DELPHI to 3PC setting. This extension employs MPHE for processing linear functions, while utilizing Garbled Circuits (GC) for non-linear operations. A key distinction between DELPHI-3 and SECO, lies in the allocation of roles within the GC scheme. In SECO, two remote servers are designated as the garbler and the evaluator in the GC process. In DELPHI-3, the gateway server assumes the role of the evaluator, while a remote server functions as the garbler. This will lead to different performance of SECO and DELPHI-3 in only online inference phase. 
% We provide a detailed description of DELPHI-3 and its comparison with SECO in Appendix~\ref{delphi3}. 
We evaluate three schemes on ResNet32 \cite{he2016deep} (with 62 layers in total). 
\paragraph{User's experience}
Figure~\ref{fig:online-time} shows the comparison of the online execution time. As illustrated in the figure, the blue bars indicate the duration for which the user must remain connected and actively process the computation while the light blue bars represent the time required to process the model on remote servers. The total height of the bars in the figure represents the user's waiting time, spanning from the start of the query to the receipt of the prediction result. We can tell that the fewer layers deployed on gateway server A, the less execution time of SECO in the preprocessing phase. When all the layers of the model are on server A, SECO is reduced to DELPHI and only the user and server A are involved in the computation. When the number of layers on server A is reduced to two, SECO offers the most efficient experience in terms of both the shortest waiting and participation time for the user comparing to DELPHI, thereby potentially minimizing the user's computational burden. We note that SECO can operate on a device with just 2 GB of memory as the user device (with less than 10 layers of ResNet32 on the gateway server), whereas DELPHI cannot support such a device, making meaningful comparison impossible. This capability allows SECO to utilize lightweight devices, such as the smart watch, as user devices. Depending on the numbers of layers on server A, SECO's results are between $2.6\times-1\times$ faster than DELPHI in terms of user's waiting time. Accordingly, SECO's efficiency ranges from $9.9\times-1\times$ faster than DELPHI in terms of user's participation time.

Figure~\ref{fig:pre-time} shows the comparison of the preprocessing execution time, revealing a similar trend as the observation in Figure~\ref{fig:online-time}. Depending on the numbers of layers on server A, SECO's results are between $11\times-1\times$ faster than DELPHI in terms of user's preprocessing time. We note that SECO's preprocessing only involves the preprocessing of $\mathbf{F}_{A}$. The three server pre-compute the secret shares for $\mathbf{F}_{BC}$ before the user appears in the system.

As illustrated in Figure~\ref{fig:online-comm} and Figure~\ref{fig:pre-comm}, SECO demonstrates up to $18.5\times$ more efficient compared to DELPHI in terms of the user's online communication cost, additionally, it is up to $16.6\times$ more efficient than DELPHI regarding the user's preprocessing communication cost. Even though SECO's total communication cost in the online inference phase is higher than DELPHI's, SECO reduces the user's cost. Figure~\ref{fig:online-comm} shows that SECO can provide a tradeoff in communication load between user and server nodes, which might be desirable to service provider.

\paragraph{Effect of schemes with different assignments of servers in GC} Figure~\ref{fig:compare1} and Figure~\ref{fig:compare2} depict the online performance of SECO, DELPHI, and DELPHI-3. We can tell from figures, the time and the cost for processing $\mathbf{F}_{A}$ are the same for SECO and DELPHI-3 as the 2PC protocol in both cases is derived from DELPHI. However, the time and cost for processing $\mathbf{F}_{BC}$ are significantly reduced in SECO. The inefficiency in DELPHI-3 arises because the evaluator of the garbled circuit obtains labels for server A's share during $\mathbf{F}_{BC}$ pre-computation, making server A's participation in the inference phase unnecessary. In contrast, SECO designates server C as the evaluator, involving only servers B and C in the inference phase. Our approach streamlines the process, avoiding the need for three-server involvement as in DELPHI-3.

\begin{figure}
\centering  
\subfigure[ResNet32 total online execution time]{
\label{fig:compare1}
\includegraphics[width=5.5cm]{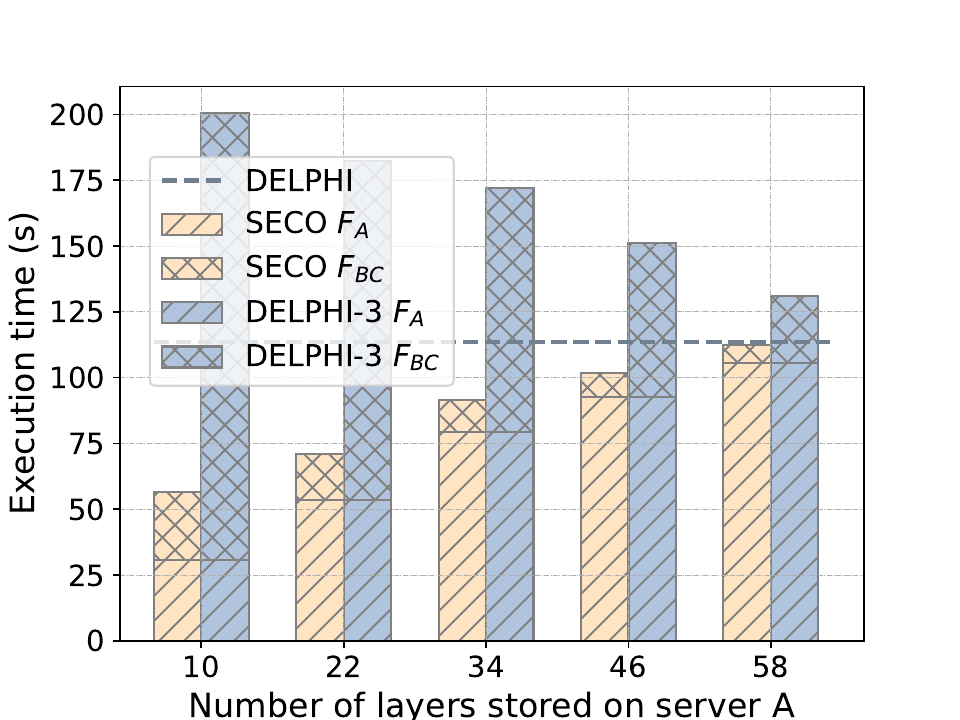}}
\subfigure[ResNet32 total online communication cost]{
\label{fig:compare2}
\includegraphics[width=5.5cm]{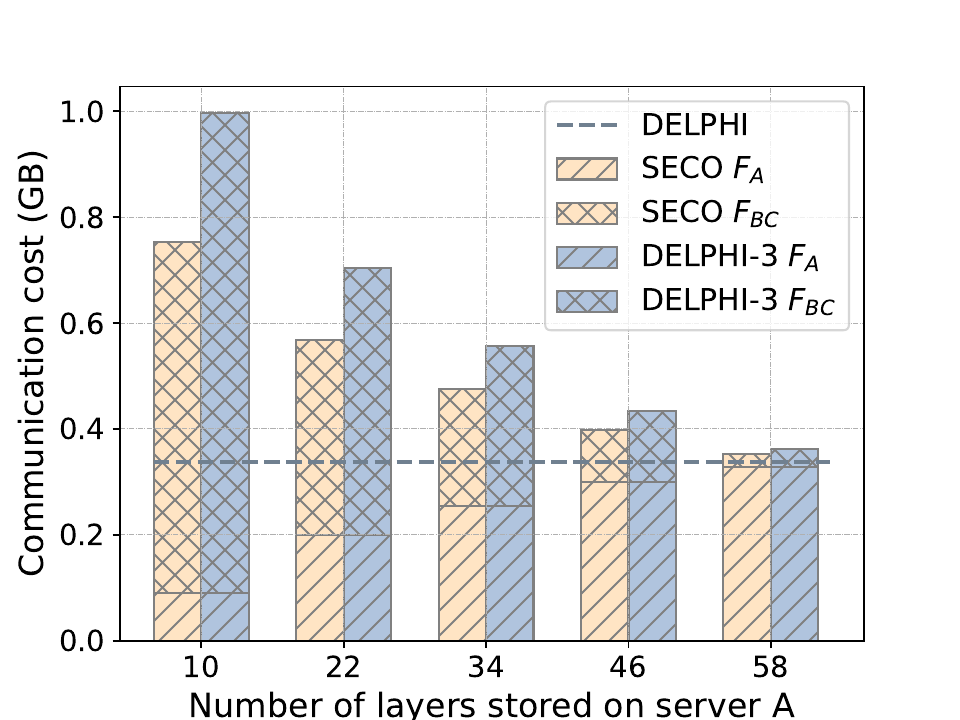}}
\caption{Comparison of SECO and DELPHI-3 on execution time and communication cost of ResNet32}
% \label{fig:communication}
\end{figure}

\subsubsection{Compare with MPC-based protocols}

In this section, we compare SECO with pure MPC-based protocols. To make a fair comparison, we deploy all the layers of the model on remote servers in SECO. So the user can send the masked input and receive the prediction result without further interaction like in pure MPC-based protocols. Figure~\ref{fig:secofalconuser} shows the communication cost for the user in protocols including SECO, Falcon \cite{wagh2020falcon}, and SecureNN \cite{wagh2019securenn}. Although there is only a one-way transmission from user to servers, due to the use of replicated secret shares, the user in Falcon needs to duplicates each of the three input shares, and the user must transmit two distinct shares to the respective server. This process leads to a communication cost that is $6\times$ greater compared to that in SECO. SecureNN uses a 2-out-of-2 secret sharing scheme between two servers, and the user needs to deliver each of two shares to the respective servers. It differs from the 2-out-of-2 secret sharing used between a user and a server in SECO. This process results in the communication cost $2\times$ greater compared to SECO. 
Furthermore, we assess the execution time starting from when the user transmits the input until the user receives the prediction result. In Table~\ref{tab:executiontime}, we evaluate 3 protocols on MiniONN and LeNet. The outperformance of SECO is attributed to the online inference only involving two (remote) servers while other protocols are executed between three servers. Although SECO is faster than other MPC-based protocols, it incurs a communication cost between servers that is 36 to 48 times higher than that of Falcon. We remind that SECO operates under a different system and threat model. SECO can protect the user's input even when the adversary corrupts two of the three servers while other protocols assume non-collusion between servers. The increase in communication cost can be justified by SECO's emphasis on enhancing the system's robustness and trustworthiness for the user.
\begin{table}[]
\centering
% \label{tab:executiontime}
\caption{Comparison of execution time and communication cost}
\label{tab:executiontime}
{\footnotesize
\begin{tabular}{ccccc}
\toprule 
         & \multicolumn{2}{c}{MiniONN} & \multicolumn{2}{c}{LeNet} \\
         \midrule
         & Time (s) & Comm. cost (MB) &Time (s) & Comm. cost \\
         \midrule
SecureNN & 183.493  & 177.264 &  -    &     -  \\
Falcon   & 8.426  & 1.097 &   9.303   & 1.611 \\
SECO     & 4.776  & 50.479 &   4.862   & 58.888\\
\bottomrule

\end{tabular}}
\end{table}

\begin{figure}[h]
    \centering
    \includegraphics[scale=0.38]{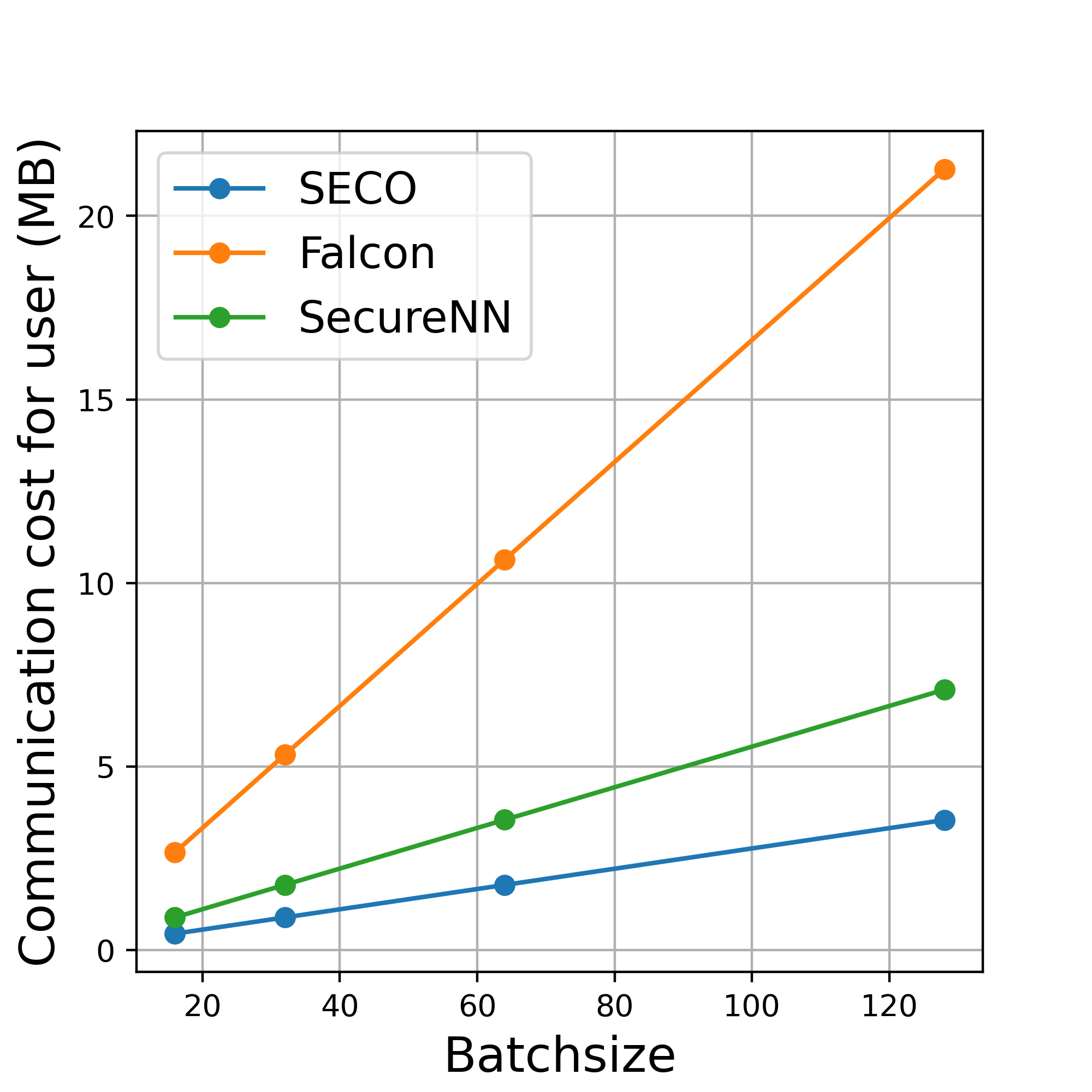} 
    \caption{User's communication cost comparison}
    \label{fig:secofalconuser}
\end{figure}
\section{Conclusion}
In this paper, we design and implement SECO, a novel hybrid HE-MPC-based secure inference protocol with model splitting in a multi-server hierarchy setting. SECO ensures the confidentiality of the input data withstanding passive adversaries corrupting up to 2 servers, and also protects the model parameters from the user. Furthermore, SECO hides more information about the model architecture than other HE-MPC-based protocols and pure MPC-based protocols. In the experiments, SECO is shown to minimize the user device's obligation in computation and communication, making it applicable to devices with limited resources.


% \section*{Acknowledgment}


\bibliographystyle{ieeetr}
\bibliography{bare_conf}
\clearpage
\appendices
\section{Notation}
Table~\ref{tab:tab1} provides the summary of the symbols used throughout this work.
\begin{table*}
\centering
\caption{Notation Summary}\label{tab:tab1}
\begin{tabular}{c|c}
\toprule
Symbol & Description \\
\midrule
$j$ & Party index (0 for the user, 1 for server A, 2 for server B, 3 for server C)\\
% $\leftarrow$& uniform sampling\\
$\mathbbm{Z}_{q}$ & $[-\frac{q}{2},\frac{q}{2})$\\
$R_t$ & Plaintext space for HE scheme\\
$R_q$ & Ciphertext space for HE scheme \\
$L$ & Total number of layers of the model  \\
$l$ & The number of layers stored on server A \\
$q_{i}$ & Input size of $i^{\text{th}}$ layer\\
$w_{i}$ & Output size of $i^{\text{th}}$ layer\\
$\mathbf{F}_{A}$ & The partial model on server A\\
$\mathbf{F}_{BC}$ & The partial model on server B and C\\
$\mathbf{F}_{i}$ & Model parameter of the $i^{\text{th}}$ layer\\
$\mathbf{F}_{i}^{j}$ & Model parameter of the $i^{\text{th}}$ layer on Party $j$\\
$\mathbf{x}_{i}$ & The input of the model's $i^{\text{th}}$ layer\\
$\hat{\mathbf{y}}$ & The prediction result\\
$\mathbf{r}_{i}$ & Randomness masking the input of $i^{\text{th}}$ layer\\
$\mathbf{s}_{i}$ & Randomness masking the output of $i^{\text{th}}$ layer\\
$(\textsf{sk}_{j},\textsf{pk}_{j})$ & Key pair generated by party $j$  \\
$\textsf{cpk}$& Common public key\\
$ct_x$ & Ciphertext of message $x$ \\
$C_i$ & Circuit for computing the $i^{\textsf{th}}$ ReLU layer \\
$\Tilde{C}_i$ &  Garbled circuit correspongind to $C_i$\\
$\textsf{label}_x$ & Garbled circuit input labels corresponding to value $x$ \\

\bottomrule
\end{tabular}
\end{table*}
\section{Security Proof} \label{sec:sec-proof}
We provide the security proof for other cases when server B and C are corrupted or server A and C are corrupted.
\paragraph{Server B, C corrupted}
The real view $v_{BC}$ of the adversary corrupting server B and C includes:
\begin{enumerate}
    \item Ciphertexts set $ct_{S_{BC}} =\{\{ct_{r_i^2},ct_{r_i^3}\}_{i=l+2}^{L},\\
    \{ct_{F_i^2},ct_{F_i^3},ct_{s_i^2},ct_{s_i^2}\}_{i=l+2}^{L},\{ct_{\textsf{share}_i^A}\}_{i=l+1}^L\}$
    \item The masked input of  $\textsf{msk}_{i}=\mathbf{x}_i-\mathbf{r}_{i}^1, i\in \{l+2,...,L\}$ 
    \item Garbled circuit $\{\Tilde{C}_i\}_{i=l+1}^L$ and $\{\textsf{label}_{i}\}_{i=l+1}^{L}$ obtained by Server B and C
\end{enumerate}
% If server B,C are corrupted, they can recover the intermediate result $\mathbf{x}_{i}, i\in \{l+1,...,L\}$ by providing the randomness respectively. However, the adversary cannot execute the white-box model inversion attack when the intermediate prediction result is present but part of its corresponding model ($\mathbf{M}_A$) is missing.
We define a real-world simulator $S_{BC}$ that simulates the view of an adversary corrupting server B and C. $S_{BC}$ is given $ \{\mathbf{F}_{i}^2, \mathbf{F}_{i}^3\}_{i=l+1}^{L}, \textsf{pk}_0$,
$ \textsf{cpk},\{(q_i, w_i)\}_{i=l+1}^{L},\textsf{msk}_{{l+1}}$ and proceeds as follows:
\begin{enumerate}
    \item $S_{BC}$ chooses a uniform random tape for server B and C.
    \item In the preprocessing phase:
    \begin{enumerate}
    \item $S_{BC}$, as server A, receives ciphertexts 
        $\textsf{MPHE.Enc}(\textsf{cpk}, \mathbf{F}_i^j),\\
    \textsf{MPHE.Enc}(\textsf{cpk}, \mathbf{r}_{i}^j), 
    \textsf{MPHE.Enc}(\textsf{cpk}, \mathbf{s}_{i}^j) $
    from server B and C.
    \item $S_{BC}$ generates $\Tilde{\mathbf{r}}_{i}^1\leftarrow \mathbbm{Z}_{Q}^{q_i}$ and computes $\widetilde{\textsf{share}}_i^A = (\mathbf{F}_i^3+\mathbf{{F}}_i^2)(\Tilde{\mathbf{r}}_{i}^1+\mathbf{r}_{i}^2+\mathbf{r}_{i}^3)-\mathbf{s}_{i}^3-{\mathbf{s}}_{i}^2$
    \item Server B sends the garbled circuit $\Tilde{C}_i$ and labels for inputs from server B to server C; $S_{BC}$ obtains the labels from server B and forwards it to server C
    \end{enumerate}
    \item In the inference phase:
    \begin{enumerate}
        \item Transition: $S_{BC}$ generates a random value as the output of the first $l$ layers $\textsf{msk}_{{l+1}}$ and sends it to server B and C
        % \item $S_{BC}$ sends $\mathbf{\Tilde{x}_{l+1}}$ to the ideal $\mathbf{M}_{BC}$ inference functionality and obtains each layer's output $\{\mathbf{\Tilde{x}_{i}}\}_{i=l+1}^{L+1}$
        % \item $S_{BC}$ obtains the labels from server B and C
        % \item $S_{BC}$ evaluates ${\Tilde{C}}_i$ and sends the output to server B and C
    \end{enumerate}
\end{enumerate}
\begin{theorem}
The view simulated by $S_{BC}$ is computaionally indistinguishable from the real view $v_{BC}$ when the adversary corrupting server B and C.
\end{theorem}
\begin{proof}
Similar to the proof of Theorem~\ref{theo:server-a}, the views are computationally indistinguishable based on the security properties of \textsf{HE}, \textsf{MPHE}, \textsf{GC}, and the use of randomness that are statistically equivalent.
\end{proof}
The arguments above yields the following security property. 
\begin{prop}
    SECO introduced in Section~\ref{section:protocol-description} securely realizes $f_{pre}$ and $f_{inf}$ in the presence of semi-honest adversary controlling server B and C.
\end{prop}
% Then we show that the view of the adversary in the real and simulated worlds are indistinguishable:
% \begin{enumerate}[itemsep=0.5pt,topsep=0.8pt,parsep=0.5pt]
%     \item The relation $\Tilde{ct}_{S_{AB}} \overset{c}{\equiv} ct_{S_{AB}}$  follows the semantic security of BFV Scheme (See Appendix~\ref{sec:securityOfBFV}).
%     \item Following the security of the Garbled Circuit (See Appendix~\ref{sec:securityOfGC}), we get 
%     \begin{eqnarray}
%     \{\widetilde{\Tilde{C}}_i,\widetilde{\textsf{label}}^B_{i},\widetilde{\textsf{label}}^C_{i}\}_{i=l+1}^{L}\overset{c}{\equiv}\{\Tilde{C}_i,\textsf{label}^B_{i},\textsf{label}^C_{i}\}_{i=1}^{L} \nonumber
%      \end{eqnarray}
     
%     \item Due to using different random values which share identical distribution hence are statistically equivalent, we get $\{\widetilde{\textsf{msk}}_{\mathbf{x}_i}\}_{i=l+2}^{L+1}\overset{c}{\equiv} \{\textsf{msk}_{\mathbf{x}_i}\}_{i=l+2}^{L+1}$
%     % \item $\{\mathbf{x}_i\}_{i=l+2}^L$ are computed by using a real intermediate prediction result $\mathbf{x}_{l+1}$ and real $\mathbf{M}_{BC}$ corresponding to a real world distribution while the simulated intermediate results $\{\Tilde{\mathbf{x}}_i\}_{i=l+2}^L$ are computed using $\mathbf{\Tilde{x}}_{l+1}$ randomly generated by $S_{BC}$ and real $\mathbf{M}_{BC}$. Without the knowledge of real input $\mathbf{x}$, we get $\mathbf{\Tilde{x}}_{i}\overset{c}{\equiv} \mathbf{x}_{i}$.
% \end{enumerate}

\paragraph{Server A, C corrupted}
% If server A, C are corrupted, not only the intermediate result $\mathbf{x}_{i}, i\in \{l+1,...,L\}$ cannot be recovered due to the missing of  $\mathbf{r}_{i}^2$ from server B, nor 
The real view $v_{AC}$ of the adversary corrupting server A and C includes:
\begin{enumerate}
    \item Ciphertexts set $ct_{S_{AC}} =\{\{ct_{r_i}\}_{i=1}^{l+1},\{ct_{r_i^2},ct_{r_i^3}\}_{i=l+2}^{L},\\\{ct_{F_i^2},ct_{F_i^3},ct_{s_i^2},ct_{s_i^2}\}_{i=l+2}^{L},\{ct_{\textsf{share}_i^A}\}_{i=l+1}^L\}$
    \item The masked input of the $i^{\text{th}}, i\in \{l+2,...,L\}$ layer $\textsf{msk}_{i}=\mathbf{x}_i-\mathbf{r}_{i}^2$
     \item The additive secret shares set $\textsf{share}^{A}_{i}=(\mathbf{F}_{i}^2+\mathbf{F}_{i}^3)\mathbf{r}_{i}-\mathbf{s}_{i}^2-\mathbf{s}_{i}^3,i\in \{l+1,...,L\}$
    \item Garbled circuit $\{\Tilde{C}_i\}_{i=1}^L$ and $\{\textsf{label}_{i}\}_{i=l+1}^{L}$ for garbled circuits
\end{enumerate}
We define a real-world simulator $S_{AC}$ that simulates the view of an adversary corrupting server A and C. $S_{AC}$ is given ${\{ \mathbf{F}_{i}^3, (q_i, w_i)\}}_{i=l+1}^L$,
$\mathbf{F}_{A},\textsf{pk}_0, \textsf{cpk}$ and proceeds as follows:
\begin{enumerate}
    \item $S_{AC}$ chooses a uniform random tape for server A and C.
    \item In the preprocessing phase:
    \begin{enumerate}
    \item $S_{AC}$, as the user, computes $\textsf{HE.Enc}(\textsf{pk}_0, \mathbf{\Tilde{r}}_{i}), i\in \{1,...,l+1\}$ to server A where $\mathbf{\Tilde{r}}_{i}$ is randomly chosen from $\mathbbm{Z}_{Q}^{q_i}$. Then $S_{AC}$ receives the evaluated ciphertext from server A.
    \item With $S_{AC}$ as the evaluator of the garbled circuits for $i^{\text{th}}, i\in \{1,...,l\}$ non-linear layer, $S_{AC}$ receives $\Tilde{C}_i, i\in \{1,...,l\}$ from server A and labels corresponding to random input generated by $S_{AC}$.
    \item $S_{AC}$, as server B, sets $\mathbf{\Tilde{F}}_i^2=\mathbf{0}$, then generates $\Tilde{\mathbf{r}}_{i}^2\leftarrow \mathbbm{Z}_{Q}^{q_i}, \Tilde{\mathbf{s}}_{i}^2\leftarrow \mathbbm{Z}_{Q}^{w_i}$ and encrypts them respectively with $\textsf{cpk}$; For processing distribute decryption, server C receives the evaluated ciphertext from server A and $\widetilde{\textsf{share}}_i^A = (\mathbf{F}_i^3+\mathbf{\Tilde{F}}_i^2)(\mathbf{r}_{i}^3+\Tilde{\mathbf{r}}_{i}^2+\mathbf{r}_{i}^1)-\mathbf{s}_{i}^3-\Tilde{\mathbf{s}}_{i}^2$ from $S_{AB}$
    \item $S_{AC}$, as server B, runs the simulator for garbled circuits and generates ${\Tilde{C}}_{i}, i\in \{l+1,...,L\}$, then sends labels corresponding to $\mathbf{\Tilde{r}}_{i}^2$ to server C
    \end{enumerate}
    \item In the inference phase:
    \begin{enumerate}
        \item Preamble: $S_{AC}$ sends $-\mathbf{\Tilde{r}}_1$ to server A to evaluate the first linear layer.
        \item $S_{AC}$, as the user, obtains labels from server A
        \item $S_{AC}$, as the user, sends a random value as the output of garbled circuits back to server A.
        \item Transition: $S_{AC}$, as server B, receives the output of the first $l$ layers
        \item $S_{AC}$, as server B, sends labels corresponding to $\mathbf{\Tilde{s}}_{i}^2$ to server C
        \item $S_{AC}$ receives the output of garbled circuits from server C
    \end{enumerate}
\end{enumerate}
\begin{theorem}
The view simulated by $S_{AC}$ is computaionally indistinguishable from the real view $v_{AC}$ when the adversary corrupting server A and C.
\end{theorem}
\begin{proof}
Similar to the proof of Theorem~\ref{theo:server-a}, the views are computationally indistinguishable based on the security properties of \textsf{HE}, \textsf{MPHE}, \textsf{GC}, and the use of randomness that are statistically equivalent.
\end{proof}
The arguments above yields the following security property. 
\begin{prop}
    SECO introduced in Section~\ref{section:protocol-description} securely realizes $f_{pre}$ and $f_{inf}$ in the presence of semi-honest adversary controlling server A and C.
\end{prop}

\section{Delphi-3} \label{delphi3}
DELPHI-3 is a 3PC protocol directly expanded from DELPHI. DELPHI-3 sets server B as the garbler and server A as the evaluator for GCs. 
\paragraph{Preprocessing for $\mathbf{F}_{BC}$} The preprocessing of linear layers is the same with SECO as described in Algorithm~\ref{Pre-com-lin-L}. The preprocessing of ReLU is as follows: with server B as the garbler, server A as the evaluator, and server C providing input as the third party, server B first prepares the garbled circuit $\Tilde{C}_i$ and labels by running \textsf{GC.Garble}$(\textsf{Params},C_i)$ where \textsf{Params} is the security parameters and $C_i$ is described in Algorithm~\ref{alg:C_i_ABC} in Appendix~\ref{sec:peseudo}. To transmit the labels corresponding to actual input values, server B sends the labels of the input value provided by itself to server A through a public channel. For input values from server A and C, server A and C first obtain the actual labels from server B via OT. Then server A collects the labels obtained by server C through a public channel. 
\paragraph{Online-inference for $\mathbf{F}_{BC}$} The inference of $\mathbf{F}_{BC}$ involves the participation of server A, B, and C. At the beginning of the inference for $i^{\text{th}}$ layer $i \in \{l+1,...,L\}$, server B and C hold $\mathbf{x}_{i}-\mathbf{r}_{i}$ where $\mathbf{r}_{i}$ is solely contributed by the user when $i=l+1$, and is contributed by server A, B and C for linear other layers of $\mathbf{F}_{BC}$. Server A, B and C collaborate
% run Algorithm~\ref{infer_L}
to do inference on $i^{\text{th}}$ layer $i \in \{l+1,...,L\}$. First server B and C will compute $\mathbf{F}_{i}^j(\mathbf{x}_{i}-\mathbf{r}_{i})+\mathbf{s}_{i}^j$, resulting in three servers holding three additive shares of the real prediction result for the current layer $(\mathbf{F}_{i}^2+\mathbf{F}_{i}^3)\mathbf{x}_{i}$. The garbled circuit $\Tilde{C}_i$ on server A (Algorithm~\ref{alg:C_i_ABC} in Appendix~\ref{sec:peseudo}) first aggregates shares to recover $(\mathbf{F}_{i}^2+\mathbf{F}_{i}^3)\mathbf{x}_{i}$, then computes ReLU of it, finally masks the ReLU output with the randomness of the next layer contributed by three servers. The entire process is in the encrypted domain. Server A can evaluate the garbled circuit when collecting all labels for input wires of $\Tilde{C}_i$. Note the inefficiency of this method is that despite three servers holding the respective additive shares of the real prediction result, the evaluator of the garbled circuit already obtains the labels for server A's share during the pre-computation for $\mathbf{F}_{BC}$. Thus server A does not need to involve in the inference phase of $\mathbf{F}_{BC}$. We observe this inefficiency and set server C as the evaluator in SECO instead.

\section{Setup phase in SECO}
The setup phase in SECO includes the MPHE key generation and the preparation for $\mathbf{F}_{BC}$. This phase involves three servers. In Table~\ref{tab:setup-1} and Table~\ref{tab:setup-2}, we provide the execution time and communication cost for the setup phase.

\begin{table*}[]
\centering
{\small
\caption{Setup phase of SECO on ResNet32 (1)}\label{tab:setup-1}
\begin{tabular}{ccccccccc}
\toprule
Numbers of layers on server B,C      & 60 & 56 & 52 & 48 & 44 & 40 & 36 &32  \\
                   \midrule
Execution time (s) &   290.202&271.366&250.965&233.168&215.994&192.566&175.395&162.568 \\
Comm. cost (GB)    &  21.609 &  19.920 &  18.230  &  16.541  &   14.851 & 13.161   &  11.747  &  10.621  \\
\bottomrule
\end{tabular}}
\end{table*}

\begin{table*}[]
\centering
{\small
\caption{Setup phase of SECO on ResNet32 (2)}\label{tab:setup-2}
\begin{tabular}{cccccccc}
\toprule
Numbers of layers on server B,C     & 28 & 24 & 20& 16 & 12 & 8 & 4 \\
                   \midrule
Execution time (s) & 148.177&136.214&117.166&91.340&68.487&41.719&16.042  \\
Comm. cost (GB)    &  9.495 & 8.368  &  7.242  &  5.822  &  4.404  &  2.985  &    1.567  \\
\bottomrule
\end{tabular}}
\end{table*}

\section{Structures of Neural Networks}
\paragraph{MiniONN} This is a 4 layer network with 2 convolutional and 2 fully-connected layers selected from prior work MiniONN \cite{liu2017oblivious}. The structure is provided in Table~\ref{minionn}. It has around 10,500 parameters in total.
\paragraph{LeNet} This network, proposed by LeCun et al. \cite{726791} contains 2 convolutional layers and 2 fully connected layers with about 431K parameters. Table~\ref{lenet} shows the structure.

\begin{table*}[h]
\centering
\begin{tabular}{c|c|c|c|c}
\toprule
\textbf{Layer} & \textbf{Type} & \textbf{Input Size} & \textbf{Output Size} & \textbf{Details} \\ \hline
1              & Convolution  & 28$\times$28$\times$1 & 24$\times$24$\times$16    & 5$\times$5 filters, stride 1 \\ \hline
2              & Subsampling  & 24$\times$24$\times$16& 12$\times$12$\times$16      & 2$\times$2 average pooling, stride 2 \\ \hline
3             & ReLU  & 1 $\times$ 2304& 1 $\times$ 2304     & \\ \hline
4              & Convolution  & 12$\times$12$\times$16 & 8$\times$8$\times$50    & 5$\times$5 filters, stride 1 \\ \hline
5              & Subsampling  & 8$\times$8$\times$16  & 4$\times$4$\times$16      & 2$\times$2 average pooling, stride 2 \\ \hline
6             & ReLU  & 1$\times$256 & 1$\times$256     &  \\ \hline
7             & Fully Connected & 256 & 100         & \\ \hline
8             & ReLU       & 100 & 100           & \\ \hline
9             & Fully Connected       & 100 & 10           & \\ \midrule
\end{tabular}
\caption{Structure of MiniONN}
\label{minionn}
\end{table*}

\begin{table*}[h]
\centering
\begin{tabular}{c|c|c|c|c}
\toprule
\textbf{Layer} & \textbf{Type} & \textbf{Input Size} & \textbf{Output Size} & \textbf{Details} \\ \hline
1              & Convolution  & 28$\times$28$\times$1 & 24$\times$24$\times$20    & 5$\times$5 filters, stride 1 \\ \hline
2              & Subsampling  & 24$\times$24$\times$20& 12$\times$12$\times$20      & 2$\times$2 average pooling, stride 2 \\ \hline
3             & ReLU  & 1 $\times$ 2880& 1 $\times$ 2880     & \\ \hline
4              & Convolution  & 12$\times$12$\times$20 & 8$\times$8$\times$50    & 5$\times$5 filters, stride 1 \\ \hline
5              & Subsampling  & 8$\times$8$\times$50  & 4$\times$4$\times$50      & 2$\times$2 average pooling, stride 2 \\ \hline
6             & ReLU  & 1$\times$800 & 1$\times$800     &  \\ \hline
7             & Fully Connected & 800 & 500         & \\ \hline
8             & ReLU       & 500 & 500           & \\ \hline
9             & Fully Connected       & 500 & 10           & \\ \midrule
\end{tabular}
\caption{Structure of LeNet}
\label{lenet}
\end{table*}

\section{Pseudocode}\label{sec:peseudo}
% \begin{algorithm}
% \caption{$\textsf{MPHE.DisDec}$} 
% \label{DistributeDec}
% {\small{
%     \begin{algorithmic}[1]
%     \Statex \textbf{Inputs}: $ct$: the ciphertext of $m$; ${\textsf{sk}_j}$: the plain secret key generated by the contributor of $\textsf{cpk}$
%     \Statex \textbf{Server A}: 
%     \Statex \hskip2.0em Transmit $ct$ to server B and C
%     \Statex \hskip2.0em $\textsf{MPHE.Reconstruct}(ct,\textsf{sk}_{1})\rightarrow pd_1$
%     \Statex \textbf{Server B and C}:
%     \Statex \hskip2.0em $\textsf{MPHE.Reconstruct}(ct,\textsf{sk}_{j}) \rightarrow pd_j$
%     \Statex \hskip2.0em Transmit $pd_j$ to server A
%     \Statex \textbf{Server A}:
%     \Statex \hskip2.0em Run $\textsf{MPHE.Dec}(ct,\{pd_j\}_{\{j \in \{1,2,3\}\}})$ to obtain $m$
%     \end{algorithmic}}}
% \end{algorithm}

\begin{algorithm}
\caption{A circuit $C_{i}$ that computes $ReLU$ function of the $i^{\text{th}}$ layer, $i \in \{1,...,l\}$} 
\label{alg:C_i_UA}
{{
    \begin{algorithmic}[1]
    \Statex Input from the user: $\mathbf{F}_{i}\mathbf{r}_{i}-\mathbf{s}_{i},\mathbf{r}_{i+1}$
    \Statex Input from server A: $\mathbf{F}_{i}(\mathbf{x}_{i}-\mathbf{r}_{i})+\mathbf{s}_{i},\mathbf{d}_{i+1}$
    \State Compute $\mathbf{F}_{i}\mathbf{x}_{i}=\mathbf{F}_{i}(\mathbf{x}_{i}-\mathbf{r}_{i})+\mathbf{s}_{i}+\mathbf{F}_{i}\mathbf{r}_{i}-\mathbf{s}_{i}$
    \State Compute $\mathbf{x}_{i+1}=ReLU(\mathbf{F}_{i}\mathbf{x}_{i})$
    \State Output $\mathbf{x}_{i+1}-\mathbf{r}_{i+1}-\mathbf{d}_{i+1}$
    \end{algorithmic}}}
\end{algorithm}
\begin{algorithm}
\caption{Pre-Compute the share for the $i^{\text{th}}$ linear layer, $i \in \{1,...l\}$} 
\label{pre-com-lin-1}
{{
    \begin{algorithmic}[1]
    \Statex \textbf{Inputs}: $(q_i, w_{i})$: the input and output shape of the $i^{\text{th}}$ layer; $(\textsf{sk}_0, \textsf{pk}_0)$: the plain HE key pair generated by the user
    \Statex \textbf{The user}: 
    \Statex \hskip2.0em $\mathbf{r}_{i}\leftarrow \mathbbm{Z}_{q}^{q_i}$ 
    \Statex \hskip2.0em Encrypt $ct_{r_i}=\textsf{HE.Enc}(\textsf{pk},\mathbf{r}_{i})$
    \Statex \hskip2.0em Transmit $ct_{r_i}$ to server A 
    \Statex \textbf{Server A}:
    \Statex \hskip2.0em $\mathbf{s}_{i} \leftarrow \mathbbm{Z}_{q}^{w_i}$
    % \Statex \hskip2.0em $\mathbf{d}_{i+1} \leftarrow \mathbbm{Z}_{q}^{q_i}$
    \Statex \hskip2.0em Compute 
    $\textsf{HE.Enc}(\textsf{pk},\mathbf{F}_{i}\mathbf{r}_{i}-\mathbf{s}_{i})$ and transmit this ciphertext to the user
    \Statex \textbf{The user}:
    \Statex \hskip2.0em Decrypt the ciphertext to obtain $\mathbf{F}_{i}\mathbf{r}_{i}-\mathbf{s}_{i}$ using $\textsf{sk}_0$
    \end{algorithmic}}}
\end{algorithm}
% \begin{algorithm}
% \caption{A circuit $C_{i}$ that computes $ReLU$ function of the $i^{\text{th}}$ layer $i \in \{l+1,...,L\}$} 
% \label{alg:C_i_ABC}
% {\small{
%     \begin{algorithmic}[1]
%     \Statex Input from server A: $(\mathbf{F}_{i}^2+\mathbf{F}_{i}^3)\mathbf{r}_{i}-\mathbf{s}_{i}^2-\mathbf{s}_{i}^3,\mathbf{r}_{i+1}^1$
%     \Statex Input from server B: $\mathbf{F}_{i}^2(\mathbf{x}_{i}-\mathbf{r}_{i})+\mathbf{s}_{i}^2, \mathbf{r}_{i+1}^2$
%     \Statex Input from server C: $\mathbf{F}_{i}^3(\mathbf{x}_{i}-\mathbf{r}_{i})+\mathbf{s}_{i}^3,\mathbf{r}_{i+1}^3$
%     \State Compute $(\mathbf{F}_{i}^2+\mathbf{F}_{i}^3)\mathbf{x}_{i}=(\mathbf{F}_{i}^2+\mathbf{F}_{i}^3)\mathbf{r}_{i}-\mathbf{s}_{i}^2-\mathbf{s}_{i}^3+\mathbf{F}_{i}^2(\mathbf{x}_{i}-\mathbf{r}_{i})+\mathbf{s}_{i}^2+\mathbf{F}_{i}^3(\mathbf{x}_{i}-\mathbf{r}_{i})+\mathbf{s}_{i}^3$
%     \State Compute $\mathbf{x}_{i+1}=ReLU((\mathbf{F}_{i}^2+\mathbf{F}_{i}^3)\mathbf{x}_{i})$
%     \State Output $\mathbf{x}_{i+1}-\mathbf{r}_{i+1}=\mathbf{x}_{i+1}-\mathbf{r}_{i+1}^1-\mathbf{r}_{i+1}^2-\mathbf{r}_{i+1}^3$
%     \end{algorithmic}}}
% \end{algorithm}

\begin{algorithm}
\caption{Pre-Compute the share for the $l+1^{\text{th}}$ linear layer} 
\label{Pre-com-lin-l-1}
{{
    \begin{algorithmic}[1]
    \Statex \textbf{Inputs}: $(q_{l+1}, w_{l+1})$: the input and output shape of the $l+1^{\text{th}}$ layer; $\textsf{cpk}$: common public key; 
    % $j$: server index ($1$ for server A, $2$ for server B, $3$ for server C)
    \Statex \textbf{The user}:
    \Statex \hskip2.0em $\mathbf{r}_{l+1}\leftarrow \mathbbm{Z}_{q}^{q_{l+1}}$
    \Statex \hskip2.0em Compute $ct_{r_{l+1}}=\textsf{MPHE.Enc}(\textsf{cpk},\mathbf{r}_{l+1})$
    \Statex \hskip2.0em Transmit $ct_{r_{l+1}}$ to server A
    \Statex \textbf{Server B and C}: 
    \Statex \hskip2.0em $\mathbf{s}_{l+1}^j\leftarrow \mathbbm{Z}_{q}^{w_{l+1}}$ 
    \Statex \hskip2.0em Compute $\textsf{MPHE.Enc}(\textsf{cpk},\mathbf{s}^j_{l+1}) \rightarrow ct_{s^j_{l+1}}$, 
    $\textsf{MPHE.Enc}(\textsf{cpk},\mathbf{F}^j_{l+1})$
    $\rightarrow ct_{F^j_{l+1}}$
    \Statex \hskip2.0em Transmit $(ct_{s^j_{l+1}},ct_{F^j_{l+1}})$ to server A 
    \Statex \textbf{Server A}:
    \Statex \hskip2.0em \textsf{MPHE.Eval}($ct_{F^2_{l+1}},ct_{F^3_{l+1}},\textsf{Add}$)$\rightarrow ct_{F_{l+1}}$
    \Statex \hskip2.0em \textsf{MPHE.Eval}($ct_{s^2_{l+1}},ct_{s^3_{l+1}},\textsf{Add}$)$\rightarrow ct_{s_{l+1}}$
    \Statex \hskip2.0em \textsf{MPHE.Eval}(\textsf{MPHE.Eval}($ct_{F_{l+1}},ct_{r_{l+1}},\textsf{Lin-OP}$)
    $,ct_{s_{l+1}},\textsf{Sub})\rightarrow ct_{F_{l+1}r_{l+1}-s_{l+1}}$
    \Statex \textbf{Server A, B, C}:
    \Statex \hskip2.0em $\textsf{MPHE.DisDec}(ct_{F_{l+1}r_{l+1}-s_{l+1}},\{sk_j\}_{j\in \{1,2,3\}})$
    \Statex \hskip2.0em Server A obtains $(\mathbf{F}^2_{l+1}+\mathbf{F}^3_{l+1})\mathbf{r}_{l+1}-\mathbf{s}^2_{l+1}-\mathbf{s}^3_{l+1}$
    \end{algorithmic}}}
\end{algorithm}
% \begin{algorithm}
% \caption{Pre-Compute the share for the $i^{\text{th}}$ linear layer, $i \in \{1,...l\}$} 
% \label{pre-com-lin-1}
% {{
%     \begin{algorithmic}[1]
%     \Statex \textbf{Inputs}: $(q_i, w_{i})$: the input and output shape of the $i^{\text{th}}$ layer; $(\textsf{sk}_0, \textsf{pk}_0)$: the plain HE key pair generated by the user
%     \Statex \textbf{The user}: 
%     \Statex \hskip2.0em $\mathbf{r}_{i}\leftarrow \mathbbm{Z}_{q}^{q_i}$ 
%     \Statex \hskip2.0em Encrypt $ct_{r_i}=\textsf{HE.Enc}(\textsf{pk},\mathbf{r}_{i})$
%     \Statex \hskip2.0em Transmit $ct_{r_i}$ to server A 
%     \Statex \textbf{Server A}:
%     \Statex \hskip2.0em $\mathbf{s}_{i} \leftarrow \mathbbm{Z}_{q}^{w_i}$
%     % \Statex \hskip2.0em $\mathbf{d}_{i+1} \leftarrow \mathbbm{Z}_{q}^{q_i}$
%     \Statex \hskip2.0em Compute 
%     $\textsf{HE.Enc}(\textsf{pk},\mathbf{F}_{i}\mathbf{r}_{i}-\mathbf{s}_{i})$ and transmit this ciphertext to the user
%     \Statex \textbf{The user}:
%     \Statex \hskip2.0em Decrypt the ciphertext to obtain $\mathbf{F}_{i}\mathbf{r}_{i}-\mathbf{s}_{i}$ using $\textsf{sk}_0$
%     \end{algorithmic}}}
% \end{algorithm}
 \begin{algorithm}
\caption{Construct Garbled Circuit for the $i^{\text{th}}$ activation layer, $i \in \{1,...,l\}$} 
\label{pre-bui-GC-1}
{{
    \begin{algorithmic}[1]
    \Statex \textbf{Inputs}: \textsf{Params}: Garbled Circuit security parameters, $C_{i}$: the circuit in Algorithm~\ref{alg:C_i_UA}; $(\mathbf{F}_{i}\mathbf{r}_{i}-\mathbf{s}_{i},\mathbf{r}_{i+1})$: the pre-computed additive share of $i^\text{th}$ layer and the randomness for next layer held by server A ; $q_i$: the input shape of $i^\text{th}$ layer
    \Statex \textbf{Server A}: 
    \Statex \hskip2.0em $\mathbf{d}_{i+1} \leftarrow \mathbbm{Z}_{q}^{q_i}$
    \Statex \hskip2.0em \textsf{GC.Garble}$(\textsf{Params},C_i)\rightarrow \Tilde{C}_{i}, \textsf{label}$
    \Statex \hskip2.0em Transmit $\Tilde{C}_{i}$ to the user
    \Statex \textbf{The user and server A}:
    \Statex \hskip2.0em \textsf{GC.Transfer}($\textsf{label},\mathbf{F}_{i}\mathbf{r}_i-\mathbf{s}_{i})\rightarrow \textsf{label}_{\mathbf{F}_{i}\mathbf{r}_i-\mathbf{s}_{i}}$
    \Statex \hskip2.0em \textsf{GC.Transfer}($\textsf{label},\mathbf{r}_{i+1})\rightarrow \textsf{label}_{\mathbf{r}_{i+1}}$
    \Statex \hskip2.0em \textsf{GC.Transfer}($\textsf{label},\mathbf{d}_{i+1})\rightarrow \textsf{label}_{\mathbf{d}_{i+1}}$
    \Statex \hskip2.0em The user obtains $\textsf{label}_{\mathbf{F}_{i}\mathbf{r}_i-\mathbf{s}_{i}},\textsf{label}_{\mathbf{r}_{i+1}},\textsf{label}_{\mathbf{d}_{i+1}}$
    \end{algorithmic}}}
\end{algorithm}

\begin{algorithm}
\caption{ Inference for the $i^{\text{th}}$ layer, $i \in \{1,...,l\}$}
\label{infer_1}
{{
    \begin{algorithmic}[1]
    \Statex \textbf{Inputs}:  $(\mathbf{F}_{i},\mathbf{x}_{i}-\mathbf{r}_{i},\mathbf{s}_{i})$: held by server A; $\Tilde{C}_{i}$: the garbled circuit of Algorithm~\ref{alg:C_i_UA}; \textsf{label}: random labels for input wires of Algorithm~\ref{alg:C_i_UA}; $\textsf{label}_{\mathbf{F}_{i}\mathbf{r}_i-\mathbf{s}_{i}},\textsf{label}_{\mathbf{r}_{i+1}},\textsf{label}_{\mathbf{d}_{i+1}}$: labels tranmitted to the user
    \Statex \textbf{Server A}: 
    \Statex \hskip2.0em Compute $\mathbf{F}_{i}(\mathbf{x}_{i}-\mathbf{r}_{i})+\mathbf{s}_{i}$
    \Statex \textbf{Server A and the user}: 
    \Statex\hskip2.0em \textsf{GC.Transfer}$(\textsf{label},$
    $\mathbf{d}_{i},\mathbf{F}_{i}(\mathbf{x}_{i}-\mathbf{r}_{i})+\mathbf{s}_{i})\rightarrow      \textsf{label}_{d_{i}}$,
                      $\textsf{label}_{F_{i}(x_{i}-r_{i})+s_{i}}$
    \Statex \hskip2.0em The user obtains $\textsf{label}_{\mathbf{F}_{i}(\mathbf{x}_{i}-\mathbf{r}_{i})+\mathbf{s}_{i}},\textsf{label}_{d_{i}}$
    \Statex \textbf{The user}:
    \Statex \hskip2.0em \textsf{GC.Eval}$(\Tilde{C}_{i},\textsf{label}_{{F}_{i}({x}_{i}-{r}_{i})+{s}_{i}},\textsf{label}_{d_{i}},\textsf{label}_{{r}_{i+1}},$
    $\textsf{label}_{{F}_{i}{r}_i-{s}_{i}})\rightarrow\mathbf{x}_{i+1}-\mathbf{r}_{i+1}-\mathbf{d}_{i+1}$
    \Statex \hskip2.0em Send $\mathbf{x}_{i+1}-\mathbf{r}_{i+1}-\mathbf{d}_{i+1}$ to Server A 
    \Statex \textbf{Server A}: 
    \Statex \hskip2.0em Compute $\mathbf{x}_{i+1}-\mathbf{r}_{i+1}$
    \end{algorithmic}}}
\end{algorithm}
% \begin{algorithm}
% \caption{$\textsf{MPHE.DisDec}$} 
% \label{DistributeDec}
% {{
%     \begin{algorithmic}[1]
%     \Statex \textbf{Inputs}: $ct$: the ciphertext of $m$; ${\textsf{sk}_j}$: the plain secret key generated by the contributor of $\textsf{cpk}$
%     \Statex \textbf{Server A}: 
%     \Statex \hskip2.0em Transmit $ct$ to server B and C
%     \Statex \hskip2.0em $\textsf{MPHE.Reconstruct}(ct,\textsf{sk}_{1})\rightarrow pd_1$
%     \Statex \textbf{Server B and C}:
%     \Statex \hskip2.0em $\textsf{MPHE.Reconstruct}(ct,\textsf{sk}_{j}) \rightarrow pd_j$
%     \Statex \hskip2.0em Transmit $pd_j$ to server A
%     \Statex \textbf{Server A}:
%     \Statex \hskip2.0em Run $\textsf{MPHE.Dec}(ct,\{pd_j\}_{\{j \in \{1,2,3\}\}})$ to obtain $m$
%     \end{algorithmic}}}
% \end{algorithm}

\begin{algorithm}
\caption{Construct the Garbled Circuits for the $i^{\text{th}}$ activation layer $i \in \{l+1,...,L\}$} 
\label{pre-bui-GC-l}
{{
    \begin{algorithmic}[1]
    \Statex \textbf{Inputs}: $C_{i}$: the circuit in Algorithm~\ref{alg:C_i_ABC}; $\mathbf{E}_{i}^1$: the pre-computed additive share held by server A; $\mathbf{r}_{i+1}^2$: the randomness generated by server B 
    \Statex \textbf{Server B}: 
    \Statex \hskip2.0em \textsf{GC.Garble}$(\textsf{Params},C_i)\rightarrow \Tilde{C}_{i},\textsf{label}$
    \Statex \hskip2.0em Transmit $\Tilde{C}_{i}$ to server C
    \Statex \textbf{Server A and B}:
    \Statex \hskip2.0em \textsf{GC.Transfer}($\textsf{label},\mathbf{E}_{i}^1)\rightarrow $
    $\textsf{label}_{\mathbf{E}_{i}^1}$
    \Statex \hskip2.0em Server A obtains $\textsf{label}_{\mathbf{E}_{i}^1}$ transmits it to server C
    \Statex \textbf{Server B and C}:
    \Statex \hskip2.0em \textsf{GC.Transfer}($\textsf{label},\mathbf{r}_{i+1}^3)\rightarrow \textsf{label}_{\mathbf{r}_{i+1}^3}$
    \Statex \hskip2.0em Server C obtains $\textsf{label}_{\mathbf{r}_{i+1}^3}$
    \end{algorithmic}}}
\end{algorithm}

\end{document}